\documentclass[10pt,a4paper,final]{article}
\usepackage[latin1]{inputenc}

\usepackage{amsmath}
\usepackage{paralist}
\usepackage{graphics} 
\usepackage{epsfig} 
\usepackage{graphicx}  \usepackage{epstopdf}

\usepackage{amsmath}
\usepackage{amsfonts}
\usepackage{amssymb}
\usepackage{graphics} 

\usepackage{hyperref}	
\usepackage{verbatim}	
\usepackage{float}		
\usepackage{subcaption}	
\usepackage[ruled,linesnumbered]{algorithm2e}
\usepackage{amsthm}
\usepackage{amssymb} 
\usepackage{amsfonts} 
\usepackage{multicol}
\usepackage{url} 

\newtheorem{definition}{Definition}
\newtheorem{theorem}{Theorem}

\providecommand{\keywords}[1]
{
	\small	
	\textbf{\textit{Keywords---}} #1
}

\date{}

\title{The Secure Link Prediction Problem}

\author{Laltu Sardar,   \\
	Cryptology and Security Research Unit,\\
	Indian Statistical Institute, Kolkata, India\\
	E-mail: laltuisical@gmail.com  
	\and
	Sushmita Ruj,\\
	R C Bose Centre for Cryptology and Security\\
	Indian Statistical Institute, Kolkata, India\\
	E-mail: sush@isical.ac.in
}

\begin{document}
	\maketitle
	
	\begin{abstract}
		Link Prediction is an important and well-studied problem for social networks. Given a snapshot of a graph, the link prediction problem predicts which new interactions between members are most likely to occur in the near future. As networks grow in size, data owners are forced to store the data in remote cloud servers which reveals sensitive information about the network. The graphs are therefore stored in encrypted form. 
		
		We study the link prediction problem on encrypted graphs. To the best of our knowledge, this secure link prediction problem has not been studied before. We use the number of common neighbors for prediction. We present three algorithms for the secure link prediction problem. We design prototypes of the schemes and formally prove their security. We execute our algorithms in real-life datasets.
	\end{abstract}


\section{Introduction}  \label{Introduction}
Social networks have become an integral part of our lives. These networks can be represented as graphs with nodes being entities (members) of the network and edges representing the association between entities (members). 
As the size of these graphs increases, it becomes quite difficult for small enterprises and business units to store the graphs in-house. So, there is a desire to store such information in cloud servers. 

In order to protect the privacy of individuals (as is now mandatory in EU and other places), data is often anonymized before storing in remote cloud servers. However, as pointed out by Backstrom \emph{et al.} \cite{BDK11}, anonymization does not imply privacy. By carefully studying the associations between members, a lot of information can be gleaned.

The data owner, therefore, has to store the data in encrypted form.
Trivially, the data owner can upload all data in encrypted form to the cloud. Whenever some query is made, data owner has to download all data, do necessary computations and re-upload the re-encrypted data. This is very inefficient and does not serve the purpose of cloud service.
Thus, we need to keep the data stored in the cloud in encrypted form in such a way that we can compute efficiently on the encrypted data.

Some basic queries for a graph are neighbor query (given a vertex return the set of vertices adjacent to it), vertex degree query (given a vertex, return the number of adjacent vertices), adjacency query (given two vertices return if there is an edge between them) etc. It is important that when an encrypted graph supports some other queries, like shortest distance queries, it should not stop supporting these basic queries.

Nowell and Kleinberg~\cite{Liben-NowellK03} first defined the link prediction problem for social networks. 
The link prediction problem states that given a snapshot of a graph whether we can predict which new interactions between members are most likely to occur in the near future.
For example, given a node $A$ at an instant, the link prediction problem tries to find the most likely node $B$ with which $A$ would like to connect at a later instant. 
Different types of distance metrics are used to measure the likelihood of the formation of new links. The distances are called \emph{score} (\cite{Liben-NowellK03}).
Nowell and Kleinberg, in \cite{Liben-NowellK03}, considered several metrics including common neighbors, Jaccard's coefficient, Adamic/Adar, preferential attachment, Katz$_\beta$ etc. For example, if $A$ and $B$ (with no edge between them) have a large number of common neighbors they are more likely to be connected in future. In this paper, for simplicity, we have considered common neighbors metric to predict the emergence of a link.

Though there has been a large body of literature on link prediction, to the best of our knowledge the secure version of the problem has not been studied to date. 
\emph{Secure Link Prediction (SLP)} problem computes link prediction algorithms over secure i.e., encrypted data.

\medskip
\noindent
\textbf{Our Contribution}\quad
We introduce the notion of secure link prediction and present three constructions. 
In particular, we ask and answer the question, ``Given a snapshot of a graph $G \equiv (V,E)$ ($V$ is the set of vertices and $E \subseteq V \times V$) at a given instant and a vertex $v\in V$, which is the most likely vertex $u$, such that, $u$ is a neighbor of $v$ at a later instant and $vu \notin E$". The score-metric we consider is the number of common neighbors of the two vertices $v$ and $u$. This can be used to answer the question, ``Given a snapshot of a graph $G=(V,E)$ at a given instant and a vertex $v\in V$, which are the $k$-most likely neighbors of $v$ at a later instant such that none of these $k$ vertices were neighbors of $v$ in $G$."

Note that the data owner outsources an encrypted copy of the graph $G$ to the cloud and sends an encrypted vertex $v$ as a query. The cloud runs the secure link prediction algorithm and returns an encrypted result, from which the client can obtain the most likely neighbor of $v$. 
The cloud knows neither the graph $G$ nor the queried vertex $v$.

It is to be noted that the client has much less computational and storage capacity. We propose three schemes, ($\mathtt{SLP}$-$\mathtt{I}$, $\mathtt{SLP}$-$\mathtt{II}$ and $\mathtt{SLP}$-$\mathtt{III}$), in all of which, the client takes the help of a proxy server which makes it efficient to obtain query results. 
At a high level: 
\begin{enumerate}
	\item $\mathtt{SLP}$-$\mathtt{I}$: is the most efficient with almost no computation at client-side and leaks only the scores to the proxy server.
	\item $\mathtt{SLP}$-$\mathtt{II}$: has a little more communication at client-side compared to $\mathtt{SLP}$-$\mathtt{I}$ but leaks the scores of a subset of vertices to the proxy server.
	\item $\mathtt{SLP}$-$\mathtt{III}$: is a very efficient scheme with almost no computation and communication at the client-side and leaks almost nothing to the proxy. This is achieved with  an extra computational and communication cost between the cloud and the proxy.  
\end{enumerate}%
In all three schemes, the client does not leak anything, to the cloud, except the number of vertices in the graph.

We have designed the scheme in such a way that it supports link prediction query as well as basic queries. Each of the previous schemes on encrypted graph are designed to support a specific query (for example, shortest distance query, focused subgraph query etc.). However, we have designed more general schemes that support not only link prediction query but also basic queries including neighbor query, vertex degree query, adjacency query etc.

All our schemes have been shown to be adaptively secure in real-ideal paradigm.

Further, we have analyzed the performance of the schemes in terms of storage requirement, computation cost and communication cost, and counted the execution time of the schemes assuming benchmark implementations of some underlying cryptographic primitives. 
we have implemented prototypes for the schemes  $\mathtt{SLP}$-$\mathtt{I}$ and $\mathtt{SLP}$-$\mathtt{II}$, and measured the performance with different real-life datasets to study the feasibility.

From the experiment, we see that they take $12.15$s and $13.75$s to encrypt whereas $8.87$s and $8.59$s process query for a graph with $10^2$ vertices.

\medskip
\noindent
\textbf{Organization}\quad
The rest of the paper is organized as follows. 
Related work is discussed in Section~\ref{sec:RelatedWorks}. 
Preliminaries and cryptographic tools are discussed in Section~\ref{sec:Preliminaries}.
Link prediction problem and its security are described in Section~\ref{sec:SLP}.
Section~\ref{sec:SLP1} describes our proposed scheme for $\mathtt{SLP}$-$\mathtt{I}$. Two improvements of $\mathtt{SLP}$-$\mathtt{I}$, $\mathtt{SLP}$-$\mathtt{II}$ and $\mathtt{SLP}$-$\mathtt{III}$, are discussed in Section~\ref{sec:SLP2} and Section~\ref{sec:SLP3} respectively.
In Section~\ref{sec:PerformanceAnalysis}, a comparative study of the complexities of our proposed schemes is given. 
In Section~\ref{sec:ExperimentalEvaluation}, details of our implementation and experimental results are shown.
A variant of link prediction problem  $\mathtt{SLP}_k$ is introduced in Section~\ref{sec:slpk}.
Finally, a summary of the paper and future research direction are given in Section~\ref{sec:Conclusion}.

\section{Related Work} \label{sec:RelatedWorks}
Graph algorithms are well studied when the graph is not encrypted. Since, necessity of outsourcing  graph data in encrypted form is increasing very fast and encryption makes it  difficult to work those algorithms, study is required to enable them.
There are only few works that deals with the `query' on `outsourced encrypted graph'. 

Chase and Kamara~\cite{ChaseK10} introduced the notion of graph encryption while they were presenting structured encryption as a generalization of searchable symmetric encryption (SSE) proposed by Song \emph{et al.}~\cite{SongWP00}. They presented schemes for \emph{neighbor queries}, \emph{adjacency queries} and \emph{focused subgraph queries} on labeled graph-structured data. 
In all of their proposed schemes, the graph was considered as an adjacency matrix and  each entry was encrypted separately using symmetric key encryption. The main idea of their scheme, given a vertex and the corresponding key, the scheme could return adjacent vertices. However, complex query requires complex operation (like addition, subtraction, division etc.) on adjacent matrix which make the scheme unsuitable.

A parallel secure computation framework \emph{GraphSC} has been designed and implemented by Nayak \emph{et al.}~\cite{NayakWIWTS15}. This framework computes functions like histogram, PageRank, matrix factorization etc. 
To run this algorithms, \emph{GraphSC} introduced parallel programming paradigms to secure computation. The parallel and secure execution enables the algorithms to perform even for large datasets. However, they adopt Path-ORAM~\cite{ccs/StefanovDSFRYD13} based techniques which is inefficient if the client has little computation power or the client doesn't uses very large size RAM.  

Sketch-based approximate shortest distance queries over encrypted graph have been studied by Meng \emph{et al.}~\cite{MengKNK15}. 
In the pre-processing stage, the client computes the sketches for every vertex that is useful for efficient shortest distance query. Instead of encrypting the graph directly, they encrypted the pre-processed data. Thus, in their scheme, there is no chance of getting information about the original graph.

Shen \emph{et al.}~\cite{ShenMZMDH18} introduced and studied cloud-based approximate \emph{constrained shortest distance queries} in encrypted graphs which finds the shortest distance with a constraint such that the total cost does not exceed a given threshold. 

Exact distance has been computed on dynamic encrypted graphs in \cite{SecGDB}. Similar to our paper, this paper uses a proxy to reduce client-side computation and information leakage to the cloud. 
In the scheme, adjacency lists are stored in an inverted index. However,  in a single query, the scheme leaks all the nodes reachable from the queried vertex which is a lot of information about the graph. For example, if the graph is complete, it reveals the whole graph.  

A graph encryption scheme, that supports top-$k$ nearest keyword search queries, has been proposed by Liu \emph{et al.}~\cite{LiuZC17}. They have made an encrypted index using order preserving encryption for searching.
Together with lightweight symmetric key encryption schemes, homomorphic encryption is used to compute on encrypted data.

Besides,
Zheng \emph{et al.}~\cite{ZhengWLH15} proposed link prediction in decentralized social network preserving the privacy. Their construction split the link score into private and public parts and applied sparse logistic regression to find links based on the content of the users. However, the graph data was not considered to be encrypted in the privacy preserving link prediction schemes.

In this paper, we outsource the graph in encrypted form. In most of the previous works, the schemes are designed to perform single specific query like neighbor query (\cite{ChaseK10}), shortest distance query (\cite{MengKNK15,ShenMZMDH18,SecGDB}), focused subgraph queries (\cite{ChaseK10}) etc. So, either it is hard to get the information about the source graph (\cite{MengKNK15}, \cite{ShenMZMDH18}), as they do not support basic queries, or leaks a lot of information for a single query (\cite{SecGDB}). One trivial approach is that taking different schemes and use all of them to support all types of required queries.
In this paper, our target is to get as much information about the graph as possible whenever required with supporting the link prediction query and leak as little information as possible. To the best of our knowledge, the secure link prediction problem has not been studied before. We study issues on link prediction problem in encrypted outsourced data and give three possible solutions overcoming them.

\section{Preliminaries} \label{sec:Preliminaries}
Let $G = (V,E)$ be a graph and $A = (a_{ij}) _{N \times N}$ be its adjacency matrix where $N$ is the number of vertices. Let $\lambda$ be the security parameter. 
Set of positive integers $\{1,2,\cdots,n\} $ is denoted by $[n]$. 
By $ x \xleftarrow{\$} X $, we mean to choose a random element from the set $X$. $D\log$ denotes the discrete logarithm. $id: \{ 0,1\}^* \rightarrow \{ 0,1\}^ {\log N}$ gives the identifiers corresponding to the vertices.
A function $negl : \mathbb{N} \leftarrow \mathbb{R}$ is said to be \emph {negligible} over $n$ if $\forall c \in \mathbb{N}$, $\exists N_c \in \mathbb{N}$ such that $\forall n> N_c,\ negl(n)<n^{-c} $. 

A probabilistic polynomial-time (PPT) permutation $\{0, 1\} ^* \times \{0, 1\} ^n \rightarrow \{0, 1\}^n$ is said to be a \emph{Pseudo Random Permutation (PRP)} if it is indistinguishable from random permutation by any PPT adversary. We consider two PRPs, $F_{k_{perm}}$ and $\pi_s$, where $k_{perm}$ and $s$ are their keys (or seeds) respectively.

\subsection{Bilinear Maps} \label{ss:BilinearMaps}
Let $\mathbb{G}$ and $\mathbb{G}_1$ be two (multiplicative) cyclic groups of order $n$ and $g$ be a generator of $\mathbb{G}$. A map $e :\mathbb{G} \times \mathbb{G} \rightarrow \mathbb{G}_1$ is said to be an \emph{admissible non-degenerate bilinear map} if--
\begin{enumerate}
	\item  $\forall u,v \in \mathbb{G}$ and $\forall a, b \in \mathbb{Z}$, we have $e(u^a,v^b) = e(u,v)^{ab}$,
	\item $ e(g,g) \neq 1$, and
	\item $e$ can be computed efficiently.
\end{enumerate}

Our algorithms use bilinear map based BGN encryption scheme \cite{2dnf}. So, we first discuss this.


\subsection{BGN Encryption Scheme} \label{ss:BGN} 
Boneh \emph{et al.}~\cite{2dnf} proposed a homomorphic encryption scheme (henceforth referred to as BGN encryption scheme) that allows an arbitrary number of additions and one multiplication.
The scheme consists of three algorithms- $\mathtt{Gen}()$, $\mathtt{Encrypt}()$ and $\mathtt{Decrypt}()$ 	. 

\begin{algorithm} \DontPrintSemicolon
	\caption{$\mathtt{Gen}(1^\lambda)$} \label{algo:bgnGen}
	
	$(q_1,\ q_2,\ \mathbb{G},\  \mathbb{G}_1,\ e) \gets \mathcal{G}(\lambda) $ \;
	$n \gets q_1q_2$	\;
	$g \xleftarrow{\$} \mathbb{G} $; $ r \xleftarrow{\$} [n] $ \;
	$ u \gets g^r $; $ h \gets u^{q_2}  $ \;
	$ sk \gets  q_1$; $ pk\gets  (n,\ \mathbb{G},\ \mathbb{G}_1,\ e,\ g,\ h) $ \;
	\Return $(pk,sk) $ \;
	
\end{algorithm}	

\medskip
\noindent
{\bf Key generation:} This takes a security parameter $\lambda$ as input and outputs a public-private key pair $(pk,sk)$ (see Algo.~\ref{algo:bgnGen}). Here, $ pk =  (n,\ \mathbb{G},\ \mathbb{G}_1,\ e,\ g,\ h) $ and $sk = q_1$. In $pk$, $e$ is a bilinear map from $\mathbb{G}\times \mathbb{G}$ to $\mathbb{G}_1$ where both $\mathbb{G} $ and $\mathbb{G}_1$ are groups of order $q_1$. Note that, given $\lambda$, $\mathcal{G}$ returns $(q_1,\ q_2,\ \mathbb{G},\  \mathbb{G}_1,\ e)$ (see \cite{2dnf}) where $q_1$ and $q_2$ are two large primes, and $\mathbb{G}$ and $\mathbb{G}_1$ are groups of order $n=q_1q_2$.

\begin{figure}[!htb]
	\centering
	\begin{minipage}{.4\textwidth}
		\begin{algorithm}[H] \DontPrintSemicolon 
			\caption{$\mathtt{Encrypt}_\mathbb{G}( pk, a )$} \label{algo:bgnEncryptG}
			
			$  (n,\ \mathbb{G},\ \mathbb{G}_1,\ e,\ g,\ h) \gets pk$  \;
			$r \xleftarrow{\$} [n]$ \;
			$c \gets g^ah^r $  \;
			\Return $c$ \;
		\end{algorithm}
	\end{minipage}%
	\begin{minipage}{.09\textwidth}
		\
	\end{minipage}%
	\begin{minipage}{0.51\textwidth}
		\begin{algorithm}[H] \DontPrintSemicolon
			\caption{$\mathtt{Decrypt}_\mathbb{G}(pk, sk,c )$} \label{algo:bgnDecryptG}
			
			$  (n,\ \mathbb{G},\ \mathbb{G}_1,\ e,\ g,\ h) \gets pk$; $q_1 \gets sk$  \;
			$c' \gets c^ {q_1}$; $\hat{g} = g^{q_1} $ \;
			$s = D\log_{\hat{g}} {c'}$ \label{dlogComputeG} \;
			\Return $s$ \;
		\end{algorithm}	
	\end{minipage}
\end{figure}


\medskip
\noindent
{\bf Encryption:} An integer $a$ is encrypted in $G$ using Algo.~\ref{algo:bgnEncryptG}. Let $a_1$ and $a_2$ be two integers that are encrypted in $\mathcal{G}$ as $c_1$ and $c_2$. Then, the bilinear map $e(c_1,c_2)$, belongs to $\mathbb{G}_1$, gives the encryption of $(a_1a_2)$. Note that arbitrary addition of plaintext is also possible in the group $\mathbb{G}_1$. If $g$ is a generator of the group $\mathbb{G}$, $e(g,g)$ acts as a generator of the group $\mathbb{G}_1$. Thus, the encryption of an integer $a$ is possible in $\mathbb{G}_1$ in similar manner (see Algo.~\ref{algo:bgnEncryptG1}).

\begin{figure}[!htb]
	\centering
	\begin{minipage}{.46\textwidth}
		\begin{algorithm}[H] \DontPrintSemicolon
			\caption{$\mathtt{Encrypt}_{\mathbb{G}_1}( pk, a )$} \label{algo:bgnEncryptG1}
			$  (n,\ \mathbb{G},\ \mathbb{G}_1,\ e,\ g,\ h) \gets pk$ \;
			$r\xleftarrow{\$} [n]$\;
			$g_1 \gets e(g,g)$; $h_1 \gets e(g,h)$ \;
			$c \gets (g_1)^a (h_1)^r $ \;
			\Return $c$ \;
		\end{algorithm}
	\end{minipage}%
	\begin{minipage}{.08\textwidth}
		\
	\end{minipage}%
	\begin{minipage}{0.46\textwidth}
		\begin{algorithm}[H] \DontPrintSemicolon
			\caption{$\mathtt{Decrypt}_{\mathbb{G}_1}(pk, sk,c )$} \label{algo:bgnDecryptG1}
			$  (n,\ \mathbb{G},\ \mathbb{G}_1,\ e,\ g,\ h) \gets pk$ \; 
			$q_1 \gets sk$ \;
			$c' \gets c^ {q_1}$; $\hat{g}_1 = e(g,g)^{q_1} $ \;
			$s = D\log_{\hat{g}} {c'}$ \label{dlogComputeG1} \;
			\Return $s$ \;
		\end{algorithm}
	\end{minipage}
\end{figure}


\medskip \noindent
{\bf Decryption:} At the time of encryption each entry is randomized. The secret key $q_1$ eliminates the randomization. Then, it is enough to find discrete logarithm $D\log$ of the rest. Algo.~\ref{algo:bgnDecryptG} and Algo.~\ref{algo:bgnDecryptG1} describes the decryption in $\mathbb{G}$ and $\mathbb{G}_1$ respectively.
In decryption algorithms, $D\log$ computation can be done with expected time $O(\sqrt{n} )$ using Pollard's lambda method \cite{MenezesOV96}. However, it can be done in constant time using some extra storage (\cite{2dnf}).

Let $\mathtt{BGN}$ be an encryption scheme as described above. Then, it is a tuple of five algorithms ($\mathtt{Gen}$, $\mathtt{Encrypt}_\mathbb{G}$, $\mathtt{Decrypt}_\mathbb{G}$, $\mathtt{Encrypt}_{\mathbb{G}_1}$, $\mathtt{Decrypt}_{\mathbb{G}_1}$) as described in Algo.~\ref{algo:bgnGen}, \ref{algo:bgnEncryptG}, \ref{algo:bgnDecryptG}, \ref{algo:bgnEncryptG1} and \ref{algo:bgnDecryptG1} respectively.

\subsection{Garbled Circuit (GC)}\label{ss:SecureComputations}
Let us consider two parties, with input $x$ and $y$ respectively, who want to compute a function $f (x, y) $. Then, a  garbled circuit~\cite{Yao82b,LindellP09}  allows them to compute $f(x,y)$ in such a way that none of the parties get any `meaningful information' about the input of the other party and none, other than the two parties, is able to compute $f(x,y)$. 

Kolesnikov \emph{et al.}~\cite{KolesnikovS08} introduced an optimization of garbled circuit that allows XOR gates to be computed without communication or cryptographic operations \cite{SecGDB}. Kolesnikov \emph{et al.}~\cite{KolesnikovSS09} presented efficient GC constructions for several basic functions using the garbled circuit construction of \cite{KolesnikovS08}. In this paper, we use garbled circuit blocks for subtraction ($\mathtt{SUB}$), comparison ($\mathtt{COMP}$) and multiplexer ($\mathtt{MUX}$) functions from \cite{KolesnikovS08}.


\section{The Secure Link Prediction (SLP) Problem} \label{sec:SLP}
Given $G = (V,E)$, let $N_{v} $ denotes the set of vertices incident on $v\in V$. 
Let $score(v,u)$ be a measure of how likely the vertex $v$ is connected to another vertex $u$ in the near future, where $vu \notin E$. 
A variant of the \emph{Link Prediction} problem states that given $v \in V$, it returns a vertex $u \in V$ ($vu \notin E$) such that $score(v,u) $ is the maximum in  $\{{ score(v,u): u \in V \setminus( N_{v} \cup \{v\} ) }\}$ i.e.,
\begin{equation}
score(v,u) \geq score(v,u'), \forall u' \in V \setminus( N_{v} \cup \{v\} )
\end{equation} 
Thus, given a vertex $v$, we find most likely vertex to connect with. There are various metrics to measure score like the number of common neighbors, Jaccard's coefficient, Adamic/Adar metric etc. 
In this paper, we consider $score(v,u)$ as the number of common nodes between $v$ and $u$ i.e., $score(v,u) = |N_{v} \cap N_{u}|$. 
Let $A$ be the adjacency matrix of the graph $G$. If $i_v$ and $i_u$ are the rows corresponding to the vertices $v$ and $u$ respectively then, the score is the inner product of the rows i.e., $score(v,u) = \sum_{k=1}^{N} A[i_v][k].A[i_u][k] $. In this paper we have used BGN encryption scheme to securely compute this inner product.    

\subsection{System Overview} 
Here, we describe the system model considered for the link prediction problem and goals which we want to achieve.

\medskip 
\noindent
{\bf System Model:} 
In our model (see Fig.~\ref{fig:systemModel}), there is a client, a cloud server, and a proxy server. Each of them communicates with others to execute the protocol.

The \emph{client} is the data owner and is considered to be \emph{trusted}. It outsources the graph in encrypted form to the cloud server and generates link prediction queries. Given a vertex $v$, it queries for the vertex $u$ which is most likely to be connected in the future. 
\begin{figure}[htbp]
	\centering
	{\includegraphics[width=0.4\textwidth]{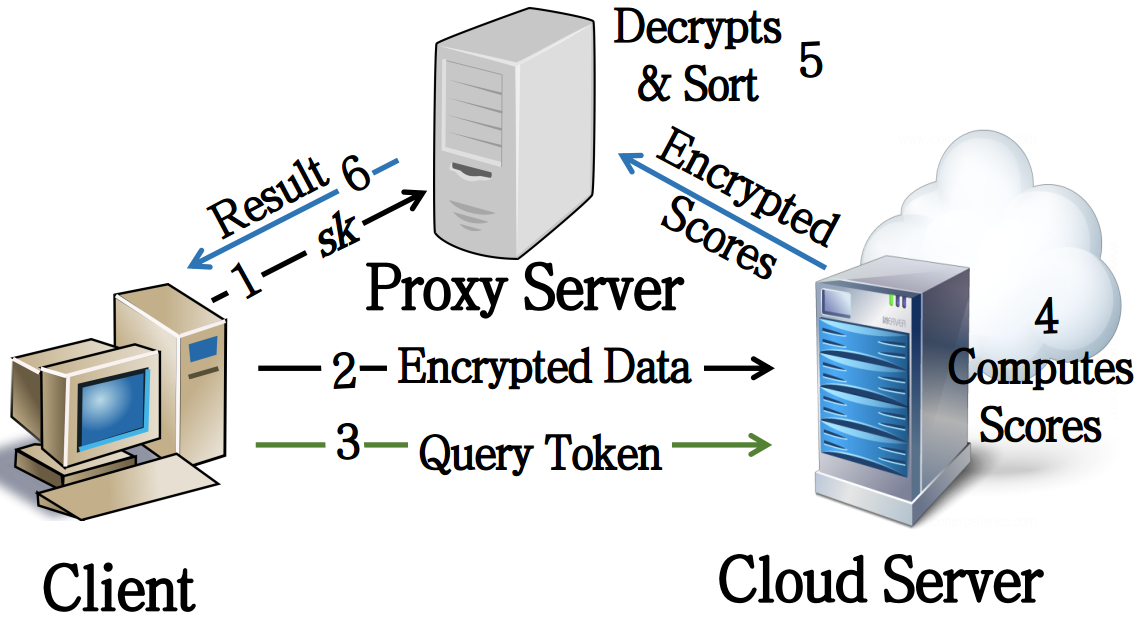}}
	\caption{System model}
	\label{fig:systemModel}
\end{figure}

The \emph{cloud server (CS)} holds the encrypted graph and computes over the encrypted data when the client requests a query. We assume that the cloud server is honest-but-curious	. It is curious to learn and analyze the encrypted data and queries. Nevertheless, it is honest and follows the protocol. 

The \emph{proxy server (PS)} helps the cloud server and the client to find the most likely vertex securely. It reduces computational overhead of the client by performing decryptions.  
However, the proxy server is assumed to be honest-but-curious.	

All channels connecting the client, the cloud and the proxy servers are assumed to be secure. An adversary can eavesdrop on channels but can not tamper messages sent along it. However, we assume, the cloud and the proxy servers do not collude.

This system model is to outsource as much computation as possible without leaking the information about the data, assuming the client has very low computation power (like mobile devices). This kind of model to outsource computation previously used by Wang et al.~\cite{SecGDB} for secure comparison. Assumption of the proxy and cloud server do not collude is a standard assumption.


\medskip 
\noindent
{\bf Design Goals:}
In this paper, under the assumption of the above system model, we aim at providing a solution for the secure link prediction problem. In our design, we want to achieve the following objectives. 
\begin{enumerate}
	\item \emph{Confidentiality:}  
	The cloud and proxy servers should not get any information about the graph structure i.e., the servers should not be able to construct a graph which is isomorphic to the source graph.  
	
	\item \emph{Efficiency:} In our model, the client is weak with respect to storage and computations. Since the cloud server has a large amount of storage and computation power, the client outsources the data to it. 
\end{enumerate}

Moreover, the client should efficiently perform neighbor query, vertex degree query or adjacency query. These are the basic query that every graph should support. The client should leak as little information as possible.

\subsection{Secure Link Prediction Scheme} \label{ss:securityDefinitions}
In this section, we present definition of link prediction scheme for a graph $G$ and its security against adaptive chosen-query attack.

\begin{definition}
	A \emph{secure link prediction ($\mathtt{SLP}$) scheme} for a graph $G$ is a tuple $(\mathtt{KeyGen} $, $\mathtt{EncMatrix} $, $\mathtt{TrapdoorGen}$, $\mathtt{LPQuery}$, $\mathtt{FindMaxVertex})$ of algorithms as follows.
	\begin{itemize} 
		\item $(\mathcal{PK}, \mathcal{SK})\gets \mathtt{KeyGen}(1^ \lambda):$ is a client-side PPT algorithm that takes $\lambda$ as a security parameter and outputs a public key $\mathcal{PK}$ and a secret key $\mathcal{SK}$.
		
		\item $ T 	\gets \mathtt{EncMatrix}(G, \mathcal{SK, PK}): $ is a client-side PPT algorithm that takes a public key $\mathcal{PK}$, a secret key $\mathcal{SK}$ and a graph $G $ as inputs and outputs a structure $T$ that stores the encrypted adjacency matrix of $G$.
		
		\item $\tau_v 	\gets \mathtt{TrapdoorGen}(v,\mathcal{SK} ):$ is a client-side PPT algorithm that takes a secret key $\mathcal{SK} $ and a vertex $v$ as inputs and outputs a query trapdoor $\tau _v$.
		
		\item $ \hat{c} \gets \mathtt{LPQuery}(\tau _v, T):$ is a PPT algorithm run by a cloud 	 server that takes a query trapdoor $\tau_v$ and the structure $T$ as inputs and outputs list of encrypted scores $\hat{c}$ with all vertices.
		
		\item $i_{res} \gets \mathtt{FindMaxVertex} (pk,sk,\hat{c}):$ is a PPT algorithm run by a proxy server that takes $pk$, $sk$ and $\hat{c}$ as inputs and outputs the most probable vertex identifier $i_{res}$ to connect with the queried vertex.
		
	\end{itemize}
\end{definition}
\noindent
{\bf Correctness:} 
An $\mathtt{SLP}$ scheme is said to be correct if,	$\forall \lambda \in \mathbb{N}$, $\forall (\mathcal{PK}, \mathcal{SK}) $ generated using $\mathtt {KeyGen}(1^\lambda)$ and all sequences of queries on $T$, each query outputs a correct vertex identifier except with negligible probability.

\bigskip \noindent
{\bf Adaptive security:} 
An $\mathtt{SLP}$ scheme should have two properties: 
\begin{enumerate}
	\item Given $T$, the cloud servers should not learn any information about $G$ and
	\item From a sequence of query trapdoors, the servers should learn nothing about corresponding queried vertices. 
\end{enumerate}

The security of an $SLP$ is defined in real-ideal paradigm.
In real scenario, the the challenger $\mathcal{C}$ generates keys. The adversary $\mathcal{A}$ generates a graph $G$ which it sends to $\mathcal{C}$. $\mathcal{C}$ encrypts the graph with its secret key and sends it to $\mathcal{A}$. Later, $q$ times it finds a query vertex based on previous results (i.e., adaptive) and receives trapdoor for the current. Finally $\mathcal{A}$ outputs a guess bit $b$.
In ideal scenario, on receiving the graph $G$, the simulator $\mathcal{S}$ generates a simulated encrypted matrix. For each adaptive query of $\mathcal{A}$, $\mathcal{S}$ returns a simulated token. Finally $\mathcal{A}$ outputs a guess bit $b'$.
The security definition (Definition~\ref{def:cqa2slp}) ensures $\mathcal{A}$ can not distinguish $\mathcal{C}$ from $\mathcal{S}$. 

We have assumed that the communication channel between the client and the servers are secure. Since the CS and the PS do not collude, they do not share their collected information. So, the simulator can treat CS and PS separately.

In our scheme, the proxy server does not have the encrypted data or the trapdoors. During query operation, it gets a set of scrambled scores of the queried vertex with other vertices. So, we can consider only the cloud server as the adversary (see \cite{BoschPLLTWHJ14}). Let us define security as follows. 

\begin{figure}[!htb]
	\centering
	\begin{minipage}{.45\textwidth}
		\begin{algorithm}[H] \DontPrintSemicolon
			\caption{ ${\textbf{Real}}^{\mathtt{SLP}}_ {\mathcal{A}}(\lambda)$} \label{algo1:gameReal}
			
			$(\mathcal{PK}, \mathcal{SK})\gets \mathtt{KeyGen}(1^ \lambda)$  \;
			$(G, st_{\mathcal{A}}) \gets \mathcal{A}_0(1^ \lambda)$ \;
			$T \gets \mathtt{EncMatrix}(G, \mathcal{SK, PK})$ \;
			$(v_1,st_{\mathcal{A}}) \gets  \mathcal{A}_1(st_{\mathcal{A}},T) $ \;
			$\tau_{v_1} \gets \mathtt{TrapdoorGen}(v_1,\mathcal{SK} )$ \;
			\For { $2 \leq i \leq q$}{
				$(v_i , st_{\mathcal{A}} ) \gets \mathcal{A}_i (st_{\mathcal{A}},T, \tau_{v_1}, \ldots , \tau_{v_{i-1}})$ \;
				$\tau_{v_i} \gets \mathtt{TrapdoorGen}(v_i,\mathcal{SK} )$ \;
			}
			$\tau = (\tau_{v_1}, \tau_{v_2}, \ldots , \tau_{v_{q}})$ \;
			$b  \gets \mathcal{A}_{q+1} { (T,\tau, st_{\mathcal{A}})}$, where $b\in \{0,1\} $ \;
			\Return $b$ \;
			
		\end{algorithm}	
	\end{minipage}%
	\begin{minipage}{.03\textwidth}
		\
	\end{minipage}%
	\begin{minipage}{0.5\textwidth}
		
		\begin{algorithm}[H] \DontPrintSemicolon
			\caption{ ${\textbf{Ideal}}^{\mathtt{SLP}}_ {\mathcal{A}, \mathcal{S}}(\lambda)$} \label{algo1:gameIdeal}
			
			$(G, st_{\mathcal{A}}) \gets \mathcal{A}_0(1^ \lambda)$  \;
			$ (st_{\mathcal{S}},T) \gets \mathcal{S}_0(\mathcal{L}_{bld}(G) ) $  \;
			$(v_1,st_{\mathcal{A}}) \gets  \mathcal{A}_1(st_{\mathcal{A}},T) $  \;
			$ (\tau_{v_1} , st_{\mathcal{S}}) \gets \mathcal{S}_1(st_{\mathcal{S}}, \mathcal{L}_{qry}(v_1)) $  \;
			\For { $2 \leq i \leq q$}{
				$(v_i , st_{\mathcal{A}} ) \gets \mathcal{A}_i (st_{\mathcal{A}},T, \tau_{v_1}, \ldots , \tau_{v_{i-1}})$  \;
				$(\tau_{v_i} , st_{\mathcal{S}}) \gets \mathcal{S}_i (st_{\mathcal{S}}, \mathcal{L}_{qry}(v_1), \ldots, \mathcal{L}_{qry}(v_{i-1}) )$  \;
			}
			$\tau = (\tau_{v_1}, \tau_{v_2}, \ldots , \tau_{v_{q}})$  \;
			$b'  \gets \mathcal{A}_{q+1} { (T, \tau, st_{\mathcal{A}})}$, where $b' \in \{0,1\} $  \;
			\Return $b'$  \;
			
		\end{algorithm}	
	\end{minipage}
\end{figure}

\begin{definition}[Adaptive semantic security $(\mathtt{CQA2})$] \label{def:cqa2slp}
	Let $\mathtt{SLP}$ = $(\mathtt{KeyGen} $, $\mathtt{EncMatrix} $, $\mathtt{TrapdoorGen}$, $\mathtt{LPQuery}$, $\mathtt{FindMaxVertex})$ be a secure link prediction scheme. Let $\mathcal{A}$ be a stateful adversary, $\mathcal{C}$ be a challenger, $\mathcal{S}$ be a stateful simulator and $\mathcal{L} = (\mathcal{L}_{bld} , \mathcal{L}_{qry} )$ be a stateful leakage algorithm. Let us consider two games- ${\textbf{Real}}^{\mathtt{SLP}}_ {\mathcal{A}}(\lambda)$ (see Algo.~\ref{algo1:gameReal}) and ${\textbf{Ideal}}^{\mathtt{SLP}}_ {\mathcal{A}, \mathcal{S}}(\lambda)$ (see Algo.~\ref{algo1:gameIdeal}).
	
	The $\mathtt{SLP}$ is said to be adaptively semantically $\mathcal{L}$-secure against chosen-query attacks ($\mathtt{CQA2}$) if, $\forall$ PPT adversaries $\mathcal{A} = (\mathcal{A}_0, \mathcal{A}_1,\ldots, \mathcal{A}_{q+1})$, where $q=poly(\lambda)$, $\exists$ a PPT simulator $\mathcal{S} = (\mathcal{S}_0, \mathcal{S}_1, \ldots, \mathcal{S}_q)$, such that 
	\begin{equation}
	|Pr[{\textbf{Real}}^{\mathtt{SLP}}_ \mathcal{A}(\lambda)=1] - Pr[{\textbf{Ideal}}^{\mathtt{SLP}} _ {\mathcal{A},\mathcal{S}}(\lambda)=1]| \leq negl(\lambda)
	\end{equation}
\end{definition}	

\subsection{Overview of our proposed schemes}
A graph can be encrypted in several ways like adjacency matrix, adjacency list, edge list etc. Each of them has advantages and disadvantages depending on the application. In our scheme, we have defined  score as the number of common neighbors that can be calculated just by computing inner product of the rows corresponding to the calculating vertices. The basic idea is that, given a vertex, to predict the most probable vertex to connect with, we compute scores with all other vertices and sort them according to their score. However, calculating the inner product and sorting them in cloud server are expensive operations and there is no scheme that provides all of the functionality to be computed over encrypted data. So, we have used BGN homomorphic encryption scheme that enables us to compute inner product on encrypted data. Choosing BGN, gives power to the client for querying not only link prediction query but also neighbor query, degree of a vertex query, adjacency query etc. 

Besides, the score computation, the score decryption and sorting the score in encrypted form is non-trivial keeping in mind that the client has low computation power.  
So, we have proposed three schemes that perform score computations as well as sorting on encrypted data with the help of a honest-but-querious proxy server which does not collude with the cloud server. The three schemes show tread-off between the computation cost, communication cost and leakage in order to compute the vertex most probable to connect with.

\section{Our Proposed Protocol for SLP}\label{sec:SLP1} 

In this section, we propose an efficient scheme $\mathtt{SLP}$-$\mathtt{I}$ and analyze its security. The scheme is divided into three phases-- key generation, data encryption, and query phase. 
The client first generates required secret and public keys. Then it encrypts the adjacency matrix of the graph in a structure and uploads it to the CS. To query for a vertex, the client generates a query trapdoor and sends it to the CS. The CS computes encrypted score (i.e., inner products of the row corresponding to the queried vertex with the other vertices on the encrypted graph). The PS decrypts the scores, finds the vertex with highest score and sends the result to the client.      


\medskip 
\noindent
{\bf Key Generation:} \label{ss4:keyGen}
In this phase, given a security parameter $\lambda$, the client chooses a bilinear map $e :\mathbb{G} \times \mathbb{G} \rightarrow \mathbb{G}_1$.   
Then, the permutation key ${k_{perm}}$ is chosen at random for the PRP $F: \{0,1\}^{*}\times \{0,1\}^{\log N} \rightarrow \{0,1\}^{\log N} $. It executes $\mathtt{BGN.Gen}() $ to get $sk$ and $pk$.
After generating private key $\mathcal{SK}$ and public key $\mathcal{PK}$, a part $sk$ of $\mathcal{SK}$ is shared with the PS. This part of the key helps the PS to compute secure comparisons. Key generation is described in Algo.~\ref{algo4:keyGen}. 

\begin{algorithm} \DontPrintSemicolon
	\caption{$\mathtt{KeyGen}(1^ \lambda)$} \label{algo4:keyGen}
	
	$ k_{perm} \xleftarrow{\$} \{0,1\}^\lambda$   \;
	$(pk,sk) \gets \mathtt{BGN.Gen}(1^\lambda) $  \;
	$\mathcal{PK} \gets pk$;
	$\mathcal{SK} \gets (sk, k_{perm} ) $  \;
	\Return $ (\mathcal{PK}, \mathcal{SK})$  \;
\end{algorithm}
\medskip \noindent
{\bf Data Encryption:} In this phase, the client encrypts the adjacency matrix with its private key and uploads the encrypted matrix to the CS (see Algo.~\ref{algo4:EncMatrix}). Each entry $a_{ij}$ in the adjacency matrix $A$ of $G$ is encrypted using Algo.~\ref{algo:bgnEncryptG}. Let $ M = (m_{ij})_{N\times N}$ be the encrypted matrix. Then, each row of $M$ is stored in the structure $T$. The PRP $F$ gives the position in $T$ corresponding to vertices. Finally, the structure $T$ is sent to the CS.

\begin{figure}[!htb]
	\centering
	\begin{minipage}{.5\textwidth}
		\begin{algorithm}[H] \DontPrintSemicolon
			\caption{$\mathtt{EncMatrixI}(A, \mathcal{SK, PK})$} \label{algo4:EncMatrix}
			
			$(n,\mathbb{G}, \mathbb{G}_1, e, g, h) \gets \mathcal{PK}$  \;
			$(q_1, k_{perm}) \gets \mathcal{SK}$  \;
			\For {$i=1, j=1$ \KwTo $i=N, j=N$} {
				$m_{ij} \gets \mathtt{BGN.Encrypt_{\mathbb{G}}}(\mathcal{PK}.pk, a_{ij}) $  \;
			} 
			Construct a structure $T$ of size $N$.   \;
			\For {$i=1$ \KwTo $i=N$} {
				$ind \gets F_{k_{perm}}(id(v_i))$  \;
				$T[ind] \gets (m_{i1},m_{i2}, \ldots, m_{iN} ) $.  \;
			} 
			\Return $T$  \;
			
		\end{algorithm}%
		
		\begin{algorithm}[H] \DontPrintSemicolon
			\caption{$\mathtt{TrapdoorGenI}(v, \mathcal{SK})$} \label{algo4:TrapdoorGen}
			
			$(sk, k_{perm} ) \gets \mathcal{SK}  $  \;
			$ i' \gets F_{k_{perm}}(id(v)) $; $s \xleftarrow{\$} \{0,1\}^{\lambda}$  \;
			$\tau_v \gets ( i' ,s)$  \;
			\Return $\tau_v$
			
		\end{algorithm}%
	\end{minipage}%
	\begin{minipage}{.02\textwidth}
		\
	\end{minipage}%
	\begin{minipage}{0.5\textwidth}
		\begin{algorithm}[H] \DontPrintSemicolon
			\caption{$\mathtt{LPQueryI}(\tau _v, T)$} \label{algo4:LPQuery}
			
			$N \gets |T|$; $(i',s ) \gets \tau_v $   \;
			$(m_{i' 1},m_{i' 2}, \ldots, m_{ i' N} ) \gets T[i'] $  \;
			\For 	{$i = 1$ \KwTo $i = N$}{
				$r \xleftarrow{\$} \{0,1\}^\lambda $  \;
				\eIf {$i \neq i' $ }{
					$(m_{i1},m_{i2}, \ldots, m_{iN} ) \gets T[{i}] $  \;
					$c_{i} \gets e(g,h)^{r} .\prod ^{N} _{k=1} e(m_{i'k},m_{ik}) $ \label{encScoreComp}  \;
				}{
					$c_{i'} \gets e(g,g)^{0}.e(g,h)^{r} $  \;
				}			
			}
			
			$\pi _s\gets$ permutation with key $s$.  \;
			$\hat{c} \gets (c_{\pi_s (1) },c_{\pi_s (2) }, \ldots, c_{\pi_s (N) })$  \; 
			$\hat{m} \gets (m_{\pi_s (1) },m_{\pi_s (2) }, \ldots, m_{\pi_s (N) })$,\\
			where $m_{i } \gets m_{i'i}.h^{r_i}$, $r_i \xleftarrow{\$} \{0,1\}^\lambda $  \;
			\Return ($ \hat{c} $, $\hat{m} $) to the PS   \;
			
		\end{algorithm}%
	\end{minipage}
\end{figure}

\medskip \noindent
{\bf Query:}
In the query phase, the client sends a query trapdoor to the CS. The CS finds encrypted scores with respect to the other vertices and sends them to the PS. The PS decrypts them and sends the identifier of  the vertex with highest score  to the client.

To query for a vertex $v$, the client first chooses a secret key $s \xleftarrow{\$} \{0,1\}^{\lambda}$ for the PRP $\pi _s$ that is not known to the PS (see  Algo.~\ref{algo4:TrapdoorGen}). Then it finds the position $ i' = F_{k_{perm}}(id(v)) $. Finally, the client sends the trapdoor $\tau_v = ( i' ,s)$ as query trapdoor to the CS.

On receiving $\tau_v $, the CS computes the encrypted scores $ (c_{1}, c_{2},\ldots , c_{N})$ (see Algo.~\ref{algo4:LPQuery}) and computes $(m_{1}, m_{2},\ldots , m_{N}) $ corresponding to the queried vertex. Using $\pi _s$, the CS shuffles the order of the encrypted scores and $m_i$'s. Finally, the CS sends the shuffled encrypted scores and the scrambled queried-row entries $ (m_{\pi _s(1) },m_{\pi_s (2) }, \ldots, m_{\pi_s (N) })$ to the PS.

\begin{algorithm} \DontPrintSemicolon
	\caption{$\mathtt{FindMaxVertexI}(sk, \hat{c} ,\hat{m} )$} \label{algo4:FindMaxVertex}
	
	$ (\bar{c}_{1 },\bar{c}_{2}, \ldots, \bar{c}_{N }) \gets \hat{c} $  \;
	$ (\bar{m}_{1 },\bar{m}_{2}, \ldots, \bar{m}_{N }) \gets \hat{m}$  \;
	\For 	{$i = 1$ \KwTo $i = N$} {
		${s}_i \gets \mathtt{BGN.Decrypt}_{\mathbb{G}_1}(pk,sk,\bar{c}_i)$  \;
		${a}_i \gets (\mathtt{BGN.Decrypt}_{\mathbb{G}}(pk, sk,\bar{m}_i)) \bmod 2$  \;
	}	  
	$i_{res} \gets i: (a_i=0) \wedge (s_i = max \{s_j:j\in [N]\})  $ \;
	\Return $i_{res} $ to the client  \;
	
\end{algorithm}

Since, the PS has $sk$ ($=q_1$), it can decrypt all $\bar {c_i}$s and $\bar {m_i}$s. It decrypts $\bar {m_i}$ first and then decrypts $\bar{c_i}$ only if corresponding decrypted value of $\bar{m_i}$ is 0. Then, it takes an ${i_{res}}$ such that $s_{i_{res}}$ is the maximum in the set $\{s_i: i \in [N] \}$ and sends it to the client (see Algo.~\ref{algo4:FindMaxVertex}). 
Finally, the client finds the resulting vertex identifier $v_{res}$ as $v_{res} \gets \pi_s ^{-1} ({i_{res}}) $.


\medskip \noindent
{\bf Correctness:}
For any two rows $T[i]$ and $T[j]$, if $c_{ij}$ is the encryption of the score $s_{ij}$ then, $c_{ij}	= e(g,h)^{r} \prod ^{N} _{k=1} e(m_{ik},m_{jk})$. Again, since ${e(g,g)^{q_1q_2}} =1$, we get $(c_{ij})^{q_1}= (e(g,g)^{q_1})^ {\sum ^{N} _{k=1} {a_{ik}a_{jk}}}$ = $\hat{g} ^{s_{ij}},\ where\ \hat{g} = e(g,g)^{q_1} $.

Thus, $D\log$ of $(c_{ij})^{q_1}$ to the base $\hat{g} $ gives $s_{ij}$. If powers of $\hat{g}$ are pre-computed, the score $s_{ij}$ can be found in constant time. However, Pollard's lambda method \cite{MenezesOV96} can be used to find discrete logarithm of $c'_{ij}$ base $\hat{g}$.


\subsection{Security Analysis} 
In the security definition, a small amount of leakage has been allowed. 
The adversary knows the algorithms and possesses the encrypted data and queried trapdoors. Only $\mathcal{SK}$ is unknown to it.
The leakage function $\mathcal{L}$ is a pair $(\mathcal{L}_{bld}, \mathcal{L}_{qry})$ (associated with $\mathtt{EncMatrix}$ and $\mathtt{LPQuery}$ respectively) where $\mathcal{L}_{bld}(G) =  \{|T|\} $ and $\mathcal{L}_{qry}(v) = \{ \tau _{v} \}$.
\begin{theorem} \label{th:security1}
	If $\mathtt{BGN}$ is semantically secure and $F$ is a PRP, then $\mathtt{SLP}$-$\mathtt{I}$ is $\mathcal{L} $-secure against adaptive chosen-query attacks.
\end{theorem}
\begin{proof} 
	The proof of security is based on the simulation-based $\mathtt{CQA}$-$\mathtt{II}$ security (see Definition~\ref{def:cqa2slp}).
	Given the leakage $\mathcal{L}_{bld}$, the simulator $\mathcal{S}$ generates a randomized structure $ \widetilde{T}$ which simulates the structure $ {T} $ of the challenger $\mathcal{C}$. 
	Given a query trapdoor $\tau_{v}$, $\mathcal{S}$ returns simulated trapdoors $\widetilde{\tau_{v}}$ maintaining system consistency of the future queries by the adversary. To prove the theorem, it is enough to 
	show that the trapdoors generated by $\mathcal{C}$ and $\mathcal{S}$ are indistinguishable to $\mathcal{A}$.
	\begin{itemize}
		\item (Simulating the structure $T$) $\mathcal{S}$ first generates $ (\mathcal{SK}, \mathcal{PK}) \gets \mathtt{BGN}.\mathtt{Gen}(1^{\lambda})$. Given $ \mathcal{L}_{bld} (A)$, $\mathcal{S}$ takes an empty structure $\widetilde{T}$ of size $|T|$. Finally, it takes $\widetilde{m_{ij}} \gets \mathtt{BGN}.\mathtt{Encrypt_\mathbb{G}}( \mathcal{PK}.pk,0^{\lambda}), \ (i, j )\in [N] \times [N]$ where $N = |T| $.
		
		\item (Simulating query trapdoor $\tau_v$) $\mathcal{S}$ first takes an empty dictionary $Q$. Given $ \mathcal{L}_{srch}(v)$, $\mathcal{S}$ checks whether $v$ is present in $Q$. If not, it takes a random $\log N$-bit string $\widetilde{ \tau_v}$, stores it as $Q[v] = \widetilde{ \tau_v} $ and returns $\widetilde{ \tau_v} $. If $v$ has appeared before, it returns $Q[v]$.		
	\end{itemize}
	Semantic security of $\mathtt{BGN}$ guarantees that $\widetilde{m_{ij}} $ and ${m_{ij}} $ are indistinguishable. Since $F$ is a PRP, $ \widetilde{ \tau_v}$ and ${ \tau_v}$ are indistinguishable. This completes the proof.
\end{proof}

\section{$\mathtt{SLP}$-$\mathtt{II}$ with less leakage} \label{sec:SLP2} 
Though the $\mathtt{SLP}$-$\mathtt{I}$ scheme is efficient, it has few disadvantages.
Firstly, in $\mathtt{SLP}$-$\mathtt{I}$, the number of common nodes between the queried vertex and all other vertices are leaked to the PS which provides partial knowledge of the graph to it.
Since, the server PS is semi honest, we want to leak as little information as possible.
In this section, we propose another scheme $\mathtt{SLP}$-$\mathtt{II}$ that hides most of the scores from the PS which results in leakage reduction. 

Secondly, the client has high communication cost with PS while processing a link prediction query.  Our proposed $\mathtt{SLP}$-$\mathtt{II}$ scheme has an advantages over this with reduced communication cost from CS to PS is. We achieve these by using extra storage of size of the matrix $M$ and extra bandwidth from the PS to the CS of $O(N)$.

\subsection{Proposed Protocol} 
In $\mathtt{SLP}$-$\mathtt{II}$, after computing the scores, the CS increases that of the incident vertices randomly from maximum possible score i.e., degree of the queried vertex. For example, let $s$ be a score in the form $g_1^s$, then a random number $r $, greater than or equal to the degree, is added with it.
Then the scores is increased as $g_1^s.g_1^r = g_1^{(s+r)}$. Since lower bound of  $r$ is known to the client, it can eliminate the scores with adjacent vertices.    
The PS only derypts the scores and sends the sorted list to the client. Since the degree is hidden from PS and known to the client, it can eliminate the vertices with score larger than degree.   
The algorithms are as follows.

\medskip
\noindent
{\bf Key Generation:} Same as Algo.~\ref{algo4:keyGen}. 

\medskip \noindent
{\bf Data Encryption:} In $\mathtt{SLP}$-$\mathtt{II}$, data encryption is similar to Algo.~\ref{algo4:EncMatrix}. 
Together with $M = (m_{ij})_{N \times N}$, another matrix $M' = (m'_{ij})_{N \times N} $ is generated by encrypting a matrix B (see Algo.~\ref{algo2:EncMatrix}). The matrix $B = (b_{ij})_{N \times N}$ where $b_{ij} = t $, ($\deg{v_i}<t<N-\deg{v_i} $) if $v_i$ and $v_j$ are connected, else $b_{ij} = 0 $. Now, $m'_{ij} = e(g,g)^{b_{ij}}.e(g,h)^{r_{ij}} $, where notations are usual. Finally, The matrices $M$ and $M'$ are uploaded to the CS together in structures $T$ and $T'$ respectively. Rows of $M$ and $M'$ corresponding to the vertex $v$ are stored in $T[F_{k_{perm}}(id(v))]$ and $T'[F_{k_{perm}}(id(v))]$ respectively.
Note that, entries of $M$ are in the group $\mathbb{G}$ whereas that of $M'$ are in $\mathbb{G}_1$.

\begin{figure}[!htb]
	\centering
	\begin{minipage}{.5\textwidth}
		\begin{algorithm}[H] \DontPrintSemicolon
			\caption{$\mathtt{EncMatrixII}(A, \mathcal{SK, PK})$} \label{algo2:EncMatrix}
			
			$(n,\mathbb{G}, \mathbb{G}_1, e, g, h) \gets \mathcal{PK}$; $(q_1, k_{perm}) \gets \mathcal{SK}$   \;
			Construct matrix $B$ from $A$  \;
			\For {$i=1, j=1$ \KwTo $i=N, j=N$} {
				$m_{ij} \gets \mathtt{BGN.Encrypt_{\mathbb{G}}}(\mathcal{PK}.pk, a_{ij}) $  \;
				$m'_{ij} \gets \mathtt{BGN.Encrypt_{\mathbb{G}_1}}(\mathcal{PK}.pk, b_{ij}) $  \;
			} 
			Construct structures $T$ and $T'$ of size $N$  \;
			\For {$i=1$ \KwTo $i=N$} {
				$ind_i \gets F_{k_{perm}}(id(v_i))$  \;
				$T[ind_i] \gets (m_{i1},m_{i2}, \ldots, m_{iN} ) $  \;
				$T'[ind_i] \gets (m'_{i1},m'_{i2}, \ldots, m'_{iN} ) $  \;
			} 
			\Return $(T, T')$  \;
		\end{algorithm}%
	\end{minipage}%
	\begin{minipage}{.02\textwidth}
		\
	\end{minipage}%
	\begin{minipage}{0.51\textwidth}
		\begin{algorithm}[H] \DontPrintSemicolon
			\caption{$\mathtt{LPQueryII}(\tau _v, T)$} \label{algo2:LPQuery}
			
			$N \gets |T|$; $(i',s ) \gets \tau_v $   \;
			$(m_{i' 1},m_{i' 2}, \ldots, m_{ i' N} ) \gets T[i'] $  \;
			\For 	{$i = 1$ \KwTo $i = N$}{
				$r \xleftarrow{\$} \{0,1\}^\lambda $  \;
				\eIf {$i \neq i' $ }{
					$(m_{i1},m_{i2}, \ldots, m_{iN} ) \gets T[{i}] $  \;
					$c_{i} \gets e(g,h)^{r} .\prod ^{N} _{k=1} e(m_{i'k},m_{ik}) $ \label{encScoreComp2}  \;
				}{
					$c_{i} \gets e(g,g)^{0}.e(g,h)^{r} $  \;
				}			
				$c_ i = c_i.m'_{i'i}$  \;
			}
			$m \gets \prod_{i = 1}^{i = N} m_{i'i}$  \;
			$\pi _s\gets$ permutation with key $s$.  \;
			$\hat{c} \gets (c_{\pi_s (1) },c_{\pi_s (2) }, \ldots, c_{\pi_s (N) })$   \;
			
			\Return $ \hat{d} $ to PS and $m$ to the client   \;
			
		\end{algorithm}%
	\end{minipage}
\end{figure}

\medskip \noindent
{\bf Query:} As in the previous scheme, the client sends query trapdoor $\tau _v = (i',s)$ to the CS for a vertex $v$. Let $ \hat{c} =(c_{1}, c_{2},\ldots , c_{N}) $ be the set of encrypted scores computed in step \ref{encScoreComp} of Algo.~\ref{algo2:LPQuery}. In addition, for each $i$, $c_ i$ is updated as $c_ i = c_i.m'_{i'i}$. Then $ \hat{c} =(c_{\pi _s (1) },c_{\pi _s (2) }, \ldots, c_{\pi _s (N) }) $ is sent to the PS. 
Instead of sending $\hat{m}$ to the PS, $m = \prod_{i = 1}^{i = N} m_{i'i}$ is sent to the client, which results the encryption of the degree of the vertex $v$. $\mathtt{SLP}$-$\mathtt{II}$ query is described in Algo.~\ref{algo2:LPQuery}.

The PS decrypts $\hat{c}$ as $s'_1, s'_2, \ldots, s'_N$ and sorts them. Then, the PS sends $ (s'_{i_1}, i_1)$, $(s'_{i_2}, i_2)$, $\ldots$, $(s'_{i_N}, i_N) $ where $s'_{i_j}$'s are in sorted order and $i_j$'s are their indices in $\hat{c}$ (see Algo.\ref{algo2:FindMaxVertex}).

The client takes the first index ${i_{res}} = i_j$ such that $s'_{i_j} \leq \deg{v}$. The client gets $\deg{v}$ by decrypting $m$. Finally, the client can find the resulting vertex identifier $v_{res}$ as $v_{res} \gets \pi _s ^{-1} ({i_{res}}) $.

\begin{algorithm} \DontPrintSemicolon
	\caption{$\mathtt{FindMaxVertexII}(sk, \bar{c} ,\bar{m} )$} \label{algo2:FindMaxVertex}
	
	$ (\bar{d}_{1 },\bar{d}_{2}, \ldots, \bar{d}_{N }) \gets \hat{d} $  \;
	\For 	{$i = 1$ \KwTo $i = N$} {
		$s'_i \gets \mathtt{BGN.Decrypt _{\mathbb{G}_1}}(pk,sk,\bar{d}_i)$ \label{s_dashed} \; 
	}	
	Sorting $s'_i$s gets ($ (s'_{i_1}, i_1)$, $(s'_{i_2}, i_2)$, $\ldots,(s'_{i_N}, i_N) $)  \;
	
	\Return ($ (s'_{i_1}, i_1)$, $(s'_{i_2}, i_2)$, $\ldots,(s'_{i_N}, i_N) $) \;
	
\end{algorithm}%


\medskip
\noindent
{\bf Correctness:}	For all $i$, the decrypted entry $s'_{i}$ (line \ref{s_dashed}, Algo.~\ref{algo2:FindMaxVertex}) is equals to $s_i + b_{i'i}$ where $s_i$ is the actual score. Since $s_i \leq \deg{v}$ and $b_{i'i}$ is zero, when $v_{i'}$ and $v_{i}$ are connected, we can see that, $s'_{i}$ becomes greater than $\deg {v}$ when $v_{i'}$ and $v_{i}$ are connected. So, the client can eliminate these entries from the list.

\subsection{Security Analysis} $\mathtt{SLP}$-$\mathtt{II}$ does not leak any extra information to the CS than $\mathtt{SLP}$-$\mathtt{I}$. The leakage $ \mathcal{L} =(\mathcal{L}_{bld}, \mathcal{L}_{qry})$ is same as it is in $\mathtt{SLP}$-$\mathtt{I}$.
\begin{theorem}
	If $\mathtt{BGN}$ is semantically secure and $F$ is a PRP, then $\mathtt{SLP}$-$\mathtt{II}$ is $\mathcal{L} $-secure against adaptive chosen-query attacks.
\end{theorem}
\begin{proof}
	As we have seen the proof of Theorem~\ref{th:security1}, The simulator requires to simulate the $ {T} $, $ {T'} $ and $\tau_v $. To simulate the structure $T'$, given $ \mathcal{L}_{bld} (A)$, $\mathcal{S}$ takes an empty structure $\widetilde{T}'$ of size $|T'|$. Finally, it takes $\widetilde{m'_{ij}} \gets \mathtt{BGN}.\mathtt{Encrypt_{\mathbb{G}_1}} ( \mathcal{PK}.pk, 0^{\lambda})$, $(i, j )\in [N] \times [N]$. Rest of the proof is similar as that of Theorem~\ref{th:security1}. 
\end{proof}

\section{ $\mathtt{SLP}$ scheme using garbled circuit ($\mathtt{SLP}$-$\mathtt{III}$)} \label{sec:SLP3}
In $\mathtt{SLP}$-$\mathtt{II}$, the PS is still able to get scores with many vertices and there is a good amount of communication cost from PS to the client. In this section, we propose  $\mathtt{SLP}$-$\mathtt{III}$ in which PS does not get any scores. Besides, the proxy needs to send only result  to the client which reduces communication overhead for the client.

\subsection{Protocol Description}
In $\mathtt{SLP}$-$\mathtt{III}$, after generating the keys, the client  encrypts the adjacency matrix of the graph and uploads it to the CS. At the same time, it shares a part of its secret key with the PS. In the query phase, the CS computes the encrypted scores on receiving query trapdoor from the client. However, it masks each score with random number selected by itself before sending them to the PS. The PS decrypts the masked scores and evaluates a garbled circuit, constructed by the CS (as described in Section~\ref{ss:mgc}),  to find the vertex with maximum score. 
Finally, the PS returns the index corresponding to the evaluated identifier of the vertex with maximum score.

\medskip
\noindent
{\bf Key Generation:} Same as Algo.~\ref{algo4:keyGen}. 


\medskip
\noindent
{\bf Data Encryption:} Same as Algo.~\ref{algo4:EncMatrix}.

\medskip
\noindent
{\bf Query:}
To query for a vertex $v$, the client generates a query trapdoor $t_v = (i',s)$ (see Algo.~\ref{algo4:TrapdoorGen}) and sends it to the CS.
On receiving  $\tau _ v$, the CS computes the encrypted scores $\ (c_1, c_2,\ldots,c_N)$. It then considers the row $T[i'] = (m_{i'1}, m_{i'2},\ldots,m_{i'N})$ corresponding to the queried vertex. Then,  with random $r_i$ and $r'_i$, it computes,
$\bar{c}_i \gets c_{\pi _s (i)}. \mathtt{BGN.Encrypt}_{\mathbb{G}_1}(\mathcal{PK}.pk, r_i)$ and $\bar{m}_i \gets m_{i'{\pi _s(i)}}. \mathtt{BGN.Encrypt}_{\mathbb{G}}(\mathcal{PK}.pk, r'_i)$, for all $i$. 
If the encrypted scores are sent directly, the PS can decrypt the scores directly as it has the partial secret key $sk$. That is why the CS chooses random $r_i$s and $r'_i$s to mask them. 

\begin{algorithm}[H] \DontPrintSemicolon
	\begin{multicols}{2}
		\caption{$\mathtt{LPQueryIII}(\tau _v, T)$} \label{algo5:LPQuery}
		
		$N \gets |T|$; $(i',s ) \gets \tau_v $   \;
		$(m_{i' 1},m_{i' 2}, \ldots, m_{ i' N} ) \gets T[i'] $  \;
		\For 	{$i = 1$ \KwTo $i = N$}{
			\eIf {$i \neq i' $ }{
				$(m_{i1},m_{i2}, \ldots, m_{iN} ) \gets T[{i}] $  \;
				$c_{i} \gets \prod ^{N} _{k=1} e(m_{\tau _vk},m_{ik}) $   \;
			}{
				$r \xleftarrow{\$} \{0,1\}^\lambda $  \;
				$c_{i'} \gets e(g,g)^{0}.e(g,h)^{r} $  \;
			}			
		}
		$\pi _s \gets$ permutation with key $s$.\;
		
		\For 	{$i = 1$ \KwTo $i = N$}{ 
			$r_i, r'_i,x_i, x'_i \xleftarrow{\$} \{0,1\}^\lambda $  \;
			$\bar{c}_i \gets c_{\pi _s (i)}.e(g,g)^{r_i}.e(g,h)^{x_i} $  \;
			$\bar{m}_i \gets m_{i'{\pi _s(i)}}. g^{r'_i}.h^{x'_i}$  \;
		}		
		
		$\hat{c} \gets (\bar{c}_{1 },\bar{c}_{2}, \ldots, \bar{c}_{N })$   \;
		$\hat{m} \gets (\bar{m}_{1 },\bar{m}_{2}, \ldots, \bar{m}_{N })$  \;
		Computes $MGC$  \;
		\Return ($ \hat{c} $, $\hat{m} $, $MGC$) to PS   \;
	\end{multicols}
\end{algorithm}

To find the vertex with highest score, the CS builds a garbled circuit $MGC$ (see Fig.~\ref{fig:maximalCircuit}) as described in Section~\ref{ss:mgc}.
The CS sends $\hat{c} = (\bar{c}_{1}, \bar{c}_{2}, \ldots, \bar{c}_{N})$ and $\hat{m} = (\bar{m}_{1}, \bar{m}_{2}, \ldots, \bar{m}_{N})$ together with a garbled circuit $MGC$. The CS-side algorithm is described in Algo.~\ref{algo5:LPQuery}.

The PS receives $\hat{c}$ and $\hat{m}$. $\forall i$, let $\bar{s}_i$ and $\bar{a}_i$ be the decryption of $\bar{c_i}$ and $\bar {m_i}$ respectively (see Algo.~\ref{algo5:FindMaxVertex}). Then, the PS evaluates $MGC$. During evaluation, the PS gives all $\bar{s}_i$s and $a_i$s and corresponding indices $i$s as input where $ a_i= (\bar{a}_i \bmod 2)$. The CS gives $r_i$s and $r''_i$s where $r''_i = (r'_i \bmod 2)$, $\forall i$ (see Section~\ref{ss:mgc}).  

\begin{algorithm} \DontPrintSemicolon
	\caption{$\mathtt{FindMaxVertexIII}(sk, \hat{c} ,\hat{m} , GC)$} \label{algo5:FindMaxVertex}
	
	$ (\bar{c}_{1 },\bar{c}_{2}, \ldots, \bar{c}_{N }) \gets \hat{c} $  \;
	$ (\bar{m}_{1 },\bar{m}_{2}, \ldots, \bar{m}_{N }) \gets \hat{m}$  \;
	\For 	{$i = 1$ \KwTo $i = N$} {
		$\bar{s}_i \gets \mathtt{BGN.Decrypt}_{\mathbb{G}_1}(pk,sk,\bar{c}_i)$  \;
		$\bar{a}_i \gets (\mathtt{BGN.Decrypt}_{\mathbb{G}}(pk,sk,\bar{m}_i))$  \;
		$a_i \gets \bar{a_i} \bmod 2$
	}	  
	Evaluates $MGC$ with $\bar{s}_i $ and $a_i$s as its inputs.\;
	$i_{res} \gets $ output of the $MGC$ evaluation  \;
	\Return $i_{res} $ to the client  \;
	
\end{algorithm}

From $MGC$, the PS gets an index $i_{res}$ which is sent to the client.
Finally, the client finds the resulting vertex identifier $v_{res}$ as $v_{res} \gets \pi_s ^{-1} ({i_{res}}) $.


\subsection{Maximum Garbled Circuit (MGC)} \label{ss:mgc}
We want minimum information to be leaked to both the servers. Without the knowledge of values, it is hard to find the maximum value because it is an iterative comparison process and requires several round of communication if we use only secure comparison. However, building a maximum garbled circuit allows cloud and proxy servers to find the maximum without knowing the value by anyone.

Kolesnikov and Schneider~\cite{KolesnikovSS09} first presented a garbled circuit that computes minimum from a set of distance. In their scheme, one party holds a set of points and the second party holds a single point. They used homomorphic encryption to compute the the distances from the single points to the set of points and sort them using the garble circuit. However, \emph{the original value of the points belongs to them were known to them}. 
In this paper, we have introduced a novel maximum garbled circuit ($MGC$) by which one party computes the maximum from a set of numbers, \emph{without the knowledge their values}, with the help of another party without leaking them to it.  
Given a set of scores $MGC$  outputs only the identity of the vertex with maximum score.


\noindent
{\bf Computing vertex with max score:} 
In $\mathtt{SLP}$-$\mathtt{III}$, the CS computes a garbled circuit $MGC$ (an example is shown in Fig.~\ref{fig:maximalCircuit}) for each query to find the maximum scored vertex identifier. Before computing $MGC$, in $\mathtt{SLP}$-$\mathtt{III}$, the PS gets $ (\bar{s}_1, \bar{s}_2,\ldots, \bar{s}_N)$ and $ (a_1, a_2,\ldots, a_N)$ (Algo.~\ref{algo5:FindMaxVertex}). The CS keeps $ (r_1, r_2,\ldots, r_N)$ and $ (r''_1, r''_2,\ldots, r''_N)$ which are used as input in $MGC$. During construction, it keeps the indices in the $MGC$ such a way that $MGC$ outputs only the index of the resulted maximum score. 

\begin{figure}[ht]
	\centering
	\includegraphics[width=0.49\textwidth]{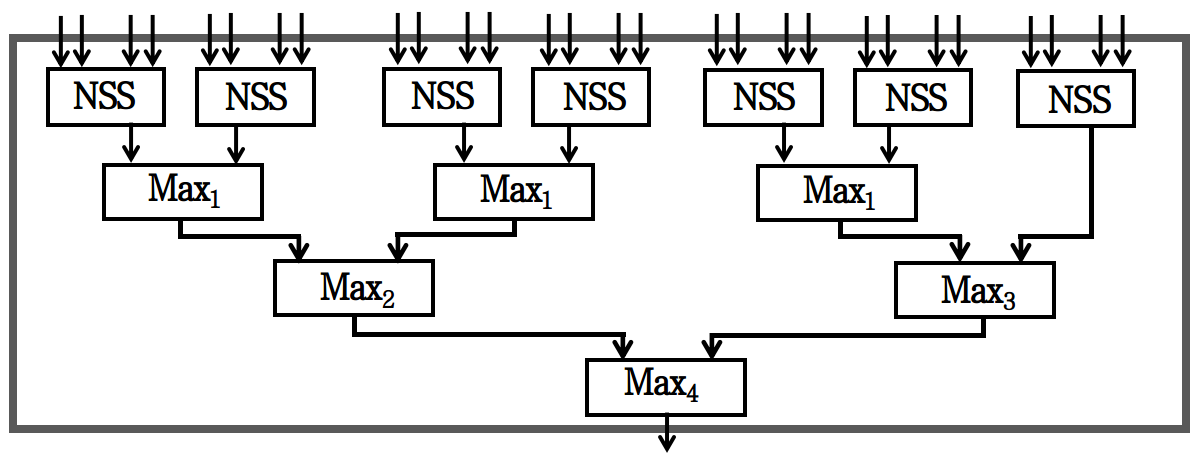}
	\caption{Example of a Maximum circuit with $N = 7$}
	\label{fig:maximalCircuit}
\end{figure} 
$MGC$ is required to find the index corresponding to the maximum scored vertex. The circuit is constructed layer by layer. The idea is to compare pair of scores every time in a layer and pass the result for the next until the resulted vertex is found. If $|V|=N$, $MGC$ has $(\log N +1 )$ layers starting from $0$ to $N$. In the 0th layer, there are $N$ number of $\mathtt{NSS}$ blocks and the rest of the blocks are $\mathtt{Max}$ block. The $\mathtt{NSS}$ blocks is for the 1st layers and computes the scores securely without knowing them. Thus, each $\mathtt{NSS}$ block corresponds to some vertex. $\mathtt{Max}$ computes the maximum score and corresponding index without knowing them.  Example of a $MGC $, to compute maximum, assuming $N = 7$ and using $\mathtt{Max}$ blocks and $\mathtt{NSS}$ blocks, is shown in Fig.~\ref{fig:maximalCircuit}. $MGC$ for any $N$ is constructed similarly.

\begin{figure}[ht]
	\centering
	\begin{subfigure}[t]{0.24\textwidth}
		\includegraphics[width=0.9\textwidth]{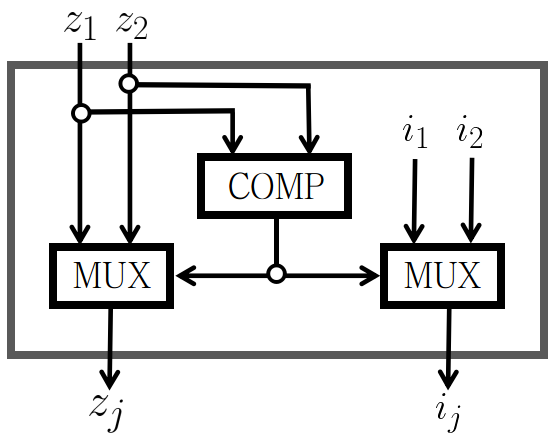}
		\caption{$\mathtt{Max}_1$ block}
		\label{fig:max1}
	\end{subfigure}%
	\begin{subfigure}[t]{0.24\textwidth}
		\includegraphics[width=0.9\textwidth]{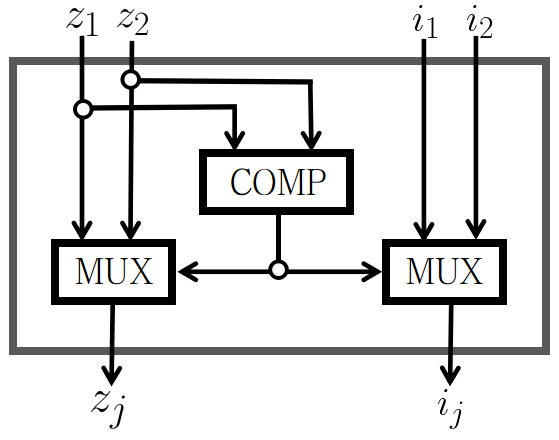}
		\caption{$\mathtt{Max}_2$ block}
		\label{fig:max2}
	\end{subfigure}
	\begin{subfigure}[t]{0.24\textwidth}
		\includegraphics[width=0.9\textwidth]{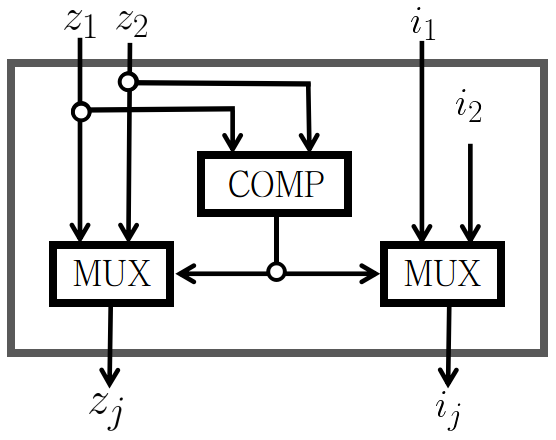}
		\caption{$\mathtt{Max}_3$ block }
		\label{fig:max3}
	\end{subfigure}%
	\begin{subfigure}[t]{0.24\textwidth}
		\includegraphics[width=0.9\textwidth]{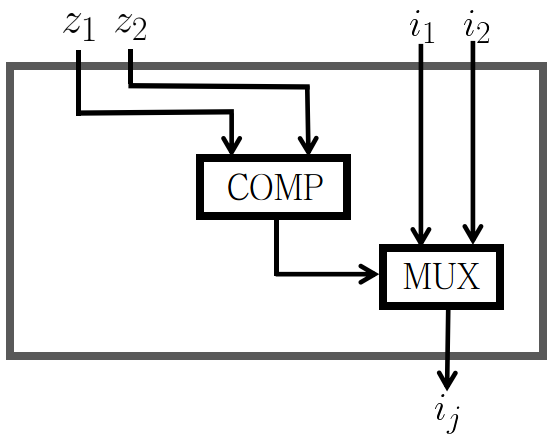}
		\caption{$\mathtt{Max}_4$ block}
		\label{fig:max4}
	\end{subfigure}%
	\caption{Different max blocks used in $\mathtt{MAXIMUM}$ circuit}
	\label{fig:maxBlocks}
\end{figure}

\noindent
{\bf Max blocks}
There are 4-types of $\mathtt{Max} $ blocks to compute the maximum- $\mathtt{Max_1}$, $\mathtt{Max_2}$, $\mathtt{Max_3}$ and $\mathtt{Max_4}$ (see Fig.~\ref{fig:maxBlocks}). The blocks are made different to handle extreme cases. These blocks use $\mathtt{COMP}$ and $\mathtt{MUX }$ blocks (see Section~\ref{ss:SecureComputations}).


\noindent
{\bf NSS blocks:}
Each $\mathtt{NSS}$ block has four inputs $\bar{s}_i$, $r_i$, $a_i$ and $r''_i$. The inputs $r_i$ and $r''_i$ comes from the CS while $\bar{s}_i$ and $a_i$ comes from the PS. It first subtracts  $r_i$ from $\bar{s}_i$ using $\mathtt{SUB}$ block to get the score $s_i$. Then, using $\mathtt{SUB}'$ block, it finds the flag bit  that tells whether the vertex is adjacent to the queried vertex.   
$\mathtt{MUL}$ block (see Fig~\ref{fig:mul}) is used in $\mathtt{NSS}$ block as shown in Fig.~\ref{fig:nss} to make the score $s_i$ zero if the vertex is adjacent else keeps the score $s_i$ same. 

\begin{figure}[ht]
	\centering
	\begin{subfigure}[t]{0.29\textwidth}
		\centering
		\includegraphics[scale=0.12]{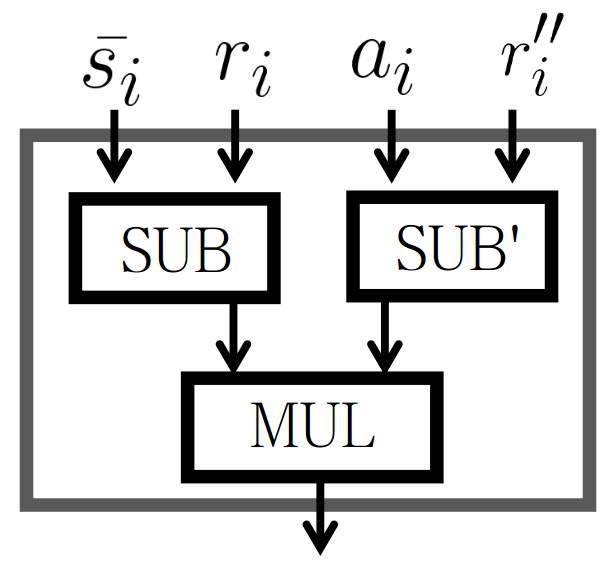}
		\caption{$\mathtt{NSS}$ block}
		\label{fig:nss}
	\end{subfigure}
	\begin{subfigure}[t]{0.44\textwidth}
		\centering
		\includegraphics[scale=0.23]{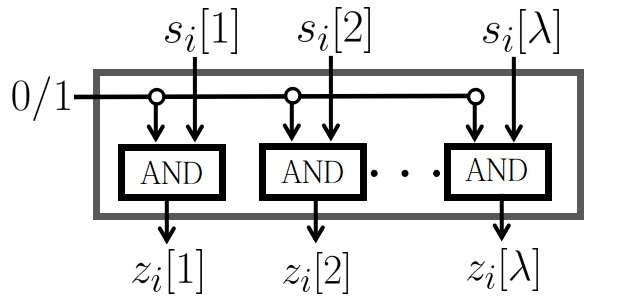}
		\caption{$\mathtt{MUL}$ block}
		\label{fig:mul}
	\end{subfigure}
	\begin{subfigure}[t]{0.24\textwidth}
		\centering
		\includegraphics[scale=0.22 ]{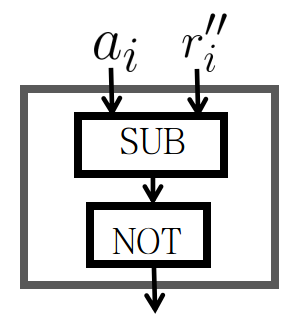}
		\caption{$\mathtt{SUB}'$ block}
		\label{fig:sub2}
	\end{subfigure}
	\caption{Few circuit blocks}
	\label{fig:fewBlocks}
\end{figure}

\noindent
{\bf Elimination of scores for adjacent vertices:}
It can be seen from encryption that $\bar{s}_i = s_i + r_i$, where $s_i$ is the actual score corresponding to $i$th row and $r_i$ randomizes the score. Each bit $r''_i$ is taken to indicate whether $r'_i$ is odd or even. On the other hand, each bit $a_i$ indicates whether the decrypted $\bar{a}_i$ is odd or even. Inequality of $r''_i$ and $a_i$ indicates that the vertex corresponding to $i$th row is connected with the queried vertex. In that case, we consider the score $s_i=0$.

The block $\mathtt{SUB}'$, in Fig.~\ref{fig:sub2}, finds outputs $1$ if they are equal, else outputs $0$. 
Since, $(\bar{s}_i - r_i) $ gives the score, $\mathtt{SUB}$ block (see Section.~\ref{ss:SecureComputations}) is used in $MGC$ to compute the scores where the PS gives $\bar{s}_i$ and CS gives $r_i$. 
It can be seen that $\mathtt{SUB}'$ subtract only one bit which is very efficient.

\subsection{Security Analysis} In $\mathtt{SLP}$-$\mathtt{III}$, though the PS has almost no leakage, the CS has a little more leakage than $\mathtt{SLP}$-$\mathtt{I}$. This extra leakage occurs when it interacts with the PS through OT protocol to provide encoding corresponding to the input of PS. Since OT is secure and does not leak any meaningful information, we can ignore this leakage. In $\mathtt{SLP}$-$\mathtt{III}$,
the leakage $ \mathcal{L} =(\mathcal{L}_{bld}, \mathcal{L}_{qry})$ is same as it is in $\mathtt{SLP}$-$\mathtt{I}$.
\begin{theorem}
	If $\mathtt{BGN}$ is semantically secure and $F$ is a PRP, then $\mathtt{SLP}$-$\mathtt{III}$ is $\mathcal{L} $-secure against adaptive chosen-query attacks.
\end{theorem}
\begin{proof}
	The proof is the same as that of Theorem~\ref{th:security1}. 
\end{proof} 

\subsection{Basic Queries}
All the three schemes support basic queries which includes neighbor query, vertex degree query and adjacency query. 

\medskip \noindent{\bf Neighbor query:}
Given a vertex,  neighbor query is to return the set of vertices adjacent to it. 
It is to be noted that, since we have encrypted adjacency matrix of the graph, it is enough for the client if it gets the decrypted row corresponding to the queried vertex, 

To query neighbor for a vertex $v$, the client generates $\tau_{v}= (i', s)$ as in Algo.~\ref{algo4:TrapdoorGen} and sends it to the CS. The CS permutes rows corresponding to row $i'$ and send the permuted row $\hat{m} \gets (m_{\pi_s (1) },m_{\pi_s (2) }, \ldots, m_{\pi_s (N) })$ to the PS. The PS decrypts them and send the decrypted vector $(a_1, a_2, \ldots, a_N)$ to the client. The client can compute inverse permutation for the entries for which the the entries are 1. Here, the CS gets only the queried vertex and the PS gets the degree of the vertex.  


\medskip \noindent{\bf  Vertex degree query:}
To query degree of a vertex $v$, similarly, the client sends $\tau_{v}= i'$ to the CS. The CS computes encrypted degree as $m \gets \prod_{i = 1}^{i = N} m_{i'i}$ and sends $m$ to the proxy. The proxy decrypts $m$ and sends the result to the client.  
$s$ is not needed as permuting the elements of some row is not required. 

Here, the degree is leaked to the PS which can be prevented by randomizing the result. The CS can randomize the encrypted degree and send the randomization secret to the client. The client can get the degree just by subtracting the randomization from the result by the PS.    

However, this leakage can be avoided easily, without randomizing the encrypted degree,  if the client performs the decryption.

\medskip \noindent{\bf Adjacency Query:}
Given two vertices, adjacency query (edge query) tells wither there is an edge between them. If the client wants to perform adjacency query for the pair of vertices $v_1$ and $v_2$, the client sends $(i'_1, i'_2)$ (as generated in Algo.~\ref{algo4:TrapdoorGen}) to the CS. The CS returns $m_{{i'_1}{ i'_2}}$. The client can get either  the randomized result from the  PS or it can decrypt $m_{{i'_1}{ i'_2}}$ by itself.

\section{Performance Analysis} \label{sec:PerformanceAnalysis}

In this section, we discuss the efficiency of our proposed schemes. The efficiency is measured in terms of computations and communication complexities together with storage requirement and allowed leakages. A summary is given in Table~\ref{tab:comparison}.
Since there is no work on the secure link prediction before, we have not compared complexities of our schemes with any other similar encrypted computations.

\subsection{Complexity analysis}
Let the graph be $G= (V,E)$ and $N = |V|$. Let $\mathtt{BGN}$ encryption outputs $\rho$-bit string for every encryption. We describe the complexities as bellow.

\medskip
\noindent {\bf Leakage Comparison:}\label{ss:LeakageComparison} 
As we see the Table~\ref{tab:comparison}, each scheme leaks, to the CS, same amount of information  which is the number of vertices of the graph and the query trapdoors. However, none of the schemes leaks information about the edges in the graph to the CS. 
In $\mathtt{SLP}$-$\mathtt{I} $, since the PS has the power to decrypt the scores, it gets to know $ S_{v} = \{score(v, u): u \in V\}$.
However, $\mathtt{SLP} $-$ \mathtt{II}$ reveals only a subset $ S'_{v} $ of $ S_{v} $ and $\mathtt{SLP}$-$\mathtt{III}$ manages to hide all scores from the PS. $\mathtt{SLP}$-$\mathtt{I}$ can not hide scores from the PS which results in maximum leakage to the PS.

\begin{table}[!htbp] 
	\centering
	\caption{Complexity Comparison Table}		\label{tab:comparison}
	\vspace{6pt}
	\resizebox{\textwidth}{!}{
		\begin{tabular}{|c|c|c|c|c|}\hline
			Param & Entity & $\mathtt{SLP}$-$\mathtt{I}$ & $\mathtt{SLP}$-$\mathtt{II}$ & $\mathtt{SLP}$-$\mathtt{III}$ \\ \hline \hline
			Leakage 	& CS & $|V|$, $\tau_{v_1},\tau_{v_2},\ldots$ &  $|V|$, $\tau_{v_1},\tau_{v_2},\ldots$  & $|V|$, $\tau_{v_1},\tau_{v_2},\ldots$  \\ \cline{2-5}
			& PS & $S_{v},i_{res}$ & $S'_{v},i_{res}$& $i_{res}$ \\ 
			\hline 
			& client & $\lambda $ bits& $\lambda$ bits& $\lambda$ bits\\ \cline{2-5}
			Storage & CS & $|V|^2\rho$ bits& $2|V|^2\rho$ bits& $|V|^2\rho$ bits\\ \cline{2-5}
			& PS & $\rho $ bits& $\rho $ bits& $\rho
			$ bits\\ 
			\hline 
			& client& $|V|^2(\mathsf{M}+\mathsf{A})$ & $|V|^2(\mathsf{M}+\mathsf{A}+\mathsf{M_1}+\mathsf{A_1})$ & $|V|^2(\mathsf{M}+\mathsf{A})$ \\ \cline{2-5}
			Compu-	& CS & $|V|^2$ $\mathsf{P}$ + $|V|$ $\mathsf{E}$ & $|V|^2$ $\mathsf{P}$ + & $|V|^2$ $\mathsf{P}$ + $4|V|$  $\mathsf{E}$ \\ 
			tation& 		& + ($|V|^2+ |V|$) $\mathsf{M}$ & ($|V|^2+ 2|V|$) $\mathsf{M}$ & + ($|V|^2+ 3|V|$) $\mathsf{M}$ +\\
			& 		& 		& 	& $MGC_{const}{(\log |V|,|V|)}$\\ \cline{2-5}
			& PS & $|V|log|V| (\mathsf{M+C}+\mathsf{M_1+C_1})$ & $|V| (\mathsf{M_1+C_1})$  +& $|V| (\mathsf{M+C}+\mathsf{M_1+C_1})$+ \\
			& 		& +$|V|log|V| \mathsf{C}$ & +$|V| log|V| \mathsf{C}$ & $MGC_{eval}{(\log |V|,|V|)}$ \\ 
			\hline 
			& client$\rightarrow$CS & $|V|^2 \rho$ bits& $2 |V|^2 \rho$ bits & $|V|^2\rho$ bits \\ \cline{2-5}
			Commu-		& CS$\rightarrow$PS& $2|V|\rho$ bits & $|V|\rho$ bits & $2|V|\rho$ bits + $|V|OT ^{(\log |V| +1)}_{snd}$+ \\
			nication& & & & $MGC_{size}{(\log |V|,|V|)}$ bits \\ \cline{2-5}
			& PS$\rightarrow$CS& - & - & $|V|OT ^{(\log |V| +1)}_{rcv}$\\ \cline{2-5}
			& PS$\rightarrow$client & $\log |V|$ bits&$2|V| \log |V|$ bits& $\log |V|$ bits\\ \hline 
			
		\end{tabular}
	}
	\begin{flushleft}
		$S_{v}$ - Set of scores of $v$ with all other vertices,
		$S'_{v} $- a subset of $ S_{v}$,
		$\rho $- length of elements in $\mathbb{G}$ or $\mathbb{G}_1$,
		$\mathsf{C}$- comparison in $\mathbb{G}$,
		$\mathsf{C_1}$- comparison in $\mathbb{G}_1$,
		$\mathsf{M}$- multiplication in $\mathbb{G}$,
		$\mathsf{M_1}$- multiplication in $\mathbb{G}_1$,
		$\mathsf{E}$- exponentiation in $\mathbb{G}$,
		$\mathsf{E_1}$- exponentiation in $\mathbb{G}_1$,
		$\mathsf{P}$- pairing/ bilinear map computation,
		$MGC_{size}{(\log |V|,|V|)}$- size of $MGC$ with $|V|$ $\log |V|$-bit inputs,
		$MGC_{const}{(\log |V|,|V|)}$- $MGC$ contraction with $|V|$ $\log |V|$-bit inputs, 
		$MGC_{eval}{(\log |V|,|V|)}$- $MGC$ evaluation with $|V|$ $\log |V|$-bit inputs,
		$OT ^{(\log |V| +1)}_{snd}$- information to send for $(\log |V|+1) $-bit $OT$,
		$OT ^{(\log |V| +1)}_{rcv}$- information to receive for $\log |V| $-bit $OT$.
	\end{flushleft}
\end{table}
\medskip
\noindent {\bf Storage Requirement:} \label{ss:StorageRequirement} 
One of the major goals of secure link prediction scheme is that the client should require very little storage. All our designed schemes have very low storage requirement for the client. The client has to only store a key which is of $\lambda$ bits. For all schemes, the PS stores only a part of the secret key which is of $\lambda$ bits. 

In $\mathtt{SLP}$-$\mathtt{I}$, the CS is required to store $|V|^2\rho$ bits for the structure $T$ where the PS is required to store only the secret key. 
While reducing the leakage in $\mathtt{SLP}$-$\mathtt{II}$, the CS storage becomes doubled. However, $\mathtt{SLP}$-$\mathtt{III}$ requires the same amount of storage as $\mathtt{SLP}$-$\mathtt{I}$.

\medskip
\noindent {\bf Computation Complexity:} \label{ss:ComputationComplexity}%
In all schemes, the client computes $|V|^2$ number of $\mathtt{BGN}$ encryption to encrypt $A$ while $\mathtt{SLP}$-$\mathtt{II}$ additionally computes $|V|^2$ number of the same to encrypt $B$. To compute each of $|V|$ encrypted scores, the CS requires $|V|$  bilinear map ($e$) computation and $|V|$  multiplications.

Additionally, $\mathtt{SLP}$-$\mathtt{I}$ randomizes the encrypted entries corresponding to the  row that has been queried. This requires $|V|$ exponentiations and $|V|$ multiplications. $\mathtt{SLP}$-$\mathtt{II}$ randomizes the encrypted scores. This requires $|V|$ multiplications and computes the encrypted degree of the queried vertex which requires $|V|$ multiplications. Apart from  computations of encrypted scores, in $\mathtt{SLP}$-$\mathtt{III}$, the CS computes a garbled circuit $MGC$.

In all, the PS decrypts $|V|$ scores. Each decryption requires $\log|V|$  multiplications on average. To find the vertex with maximum score, in $\mathtt{SLP}$-$\mathtt{I}$ and $\mathtt{SLP}$-$\mathtt{II}$, the PS compares $|V|$ numbers. The $|V|$ encrypted entries are decrypted by the PS in $\mathtt{SLP}$-$\mathtt{I}$ and $\mathtt{SLP}$-$\mathtt{III}$. In addition, the PS evaluates the garbled circuit $MGC$ in $\mathtt{SLP}$-$\mathtt{III}$.

\medskip
\noindent {\bf Communication Complexity:} \label{ss:CommunicationComplexity}
To upload the encrypted matrices, $\mathtt{SLP}$-$\mathtt{I}$ and $\mathtt{SLP}$-$\mathtt{III}$ requires $ |V|^2\rho$ bits and $\mathtt{SLP}$-$\mathtt{II}$ requires $ 2|V|^2\rho$ bits of communications. To query, it sends only the trapdoor of size $2\rho$ bits (aprx.). 

The CS sends $2|V|$ entries to the PS, in case of $\mathtt{SLP}$-$\mathtt{I}$ and $\mathtt{SLP}$-$\mathtt{III}$. For $\mathtt{SLP}$-$\mathtt{II}$, the CS sends only $|V|$ entries. Each of these entries is of $\rho$ bits. In addition, $\mathtt{SLP}$-$\mathtt{III}$ sends the garbled circuit $MGC$.
PS to CS communication happens only when the PS evaluates $MGC$. 
For $\mathtt{SLP}$-$\mathtt{I}$ and $\mathtt{SLP}$-$\mathtt{III}$, the PS sends only $i_{res}$ which is of $\log |V|$ bits to the client. However, the PS sends $2|V| \log |V|$ bits to the client. 

\medskip
\noindent {\bf Complexity for GC Computation:} 
It can be observed that $\log |V|$-bit $\mathtt{SUB}$, $1$-bit $\mathtt{SUB}'$, $\log |V|$-bit $\mathtt{MUL}$, $\log |V|$-bit $\mathtt{COMP}$ and $\log |V|$-bit $\mathtt{MUX}$ blocks consist of ($4\log |V|$ XOR-gates and $\log |V|$ AND-gates), ($4$ XOR-gates and $1$ AND-gate), ($\log |V|$ AND-gates), ($3\log |V|$ XOR-gates and $\log |V|$ AND-gates) and ($2\log |V|$ XOR-gates and $\log |V|$ AND-gates) respectively. Thus, $\log |V|$-bit $\mathtt{NSS}$ and $\log |V|$-bit $\mathtt{Max}$ blocks consist of ($(4\log |V|+4)$ XOR-gates and $(2\log |V|+1)$ AND-gates) and ($7\log |V|$ XOR-gates and $3\log |V|$ AND-gates) respectively.

In our designed garbled circuit $MGC$, there are $(|V|-1)$ $\mathtt{Max}$ blocks and $|V|$ $\mathtt{NSS}$ blocks. Thus, $MGC$ requires $|V|(11\log |V|+4)$ XOR-gates and $|V|(5\log |V|+1)$ AND-gates. However, the PS receives $|V|(\log |V| +1)$ bits through OT for the first layer. 

Thus, $MGC_{size}{(\log |V|,|V|)}$ is the size of $|V|(11\log |V|+4)$ XOR-gates and $|V|(5\log |V|+1)$ AND-gates, whereas 
$MGC_{const}{(\log |V|,|V|)}$ and $MGC_{eval}{(\log |V|,|V|)}$ are computational cost to construct and evaluate.

\section{Experimental Evaluation} \label{sec:ExperimentalEvaluation}
In this section, the experimental evaluations of our designed schemes, $\mathtt{SLP}$-$\mathtt{I}$ and $\mathtt{SLP}$-$\mathtt{II}$, are presented. 
In our experiment, we  have used a single machine for both the client and the server. All data has been assumed to be residing in main memory. The machine is with an Intel Core i7-4770 CPU and with 8-core operating at 3.40GHz. It is equipped with 8GB RAM and runs an Ubuntu 16.04 LTS 64-bit operating system. The open source PBC~\cite{PBC} library has been used in our implementation to support $\mathtt{BGN}$. The code is in the repository \cite{slpImp}.

\subsection{Datasets}
For our experiment, we have used real-world datasets. We have taken the datasets from the \emph{SNAP datasets}~\cite{snapnets}. The collection consists of various kinds of real-world network data which includes social networks, citation networks, collaboration networks, web graphs etc. 

\begin{table}[ht]
	\caption{Detail of the graph datasets} \label{tab:dataset} \centering
	\begin{tabular}{l|r|r}
		\textbf{Dataset Name}	& \textbf{\#Nodes} 	& \textbf{\#Edges} \\ \hline \hline
		bitcoin-alpha		& 3,783		& 24,186	\\ \hline
		ego-facebook 		& 4,039		& 88,234	\\ \hline
		email-Enron			& 36,692	& 183,831	\\ \hline
		email-Eu-core		& 1,005		& 25,571	\\ \hline
		Wiki-Vote			& 7,115		& 103,689	\\ 
	\end{tabular}
\end{table}

For our experiment, we have considered the undirected graph datasets- 
\emph{bitcoin-alpha},
\emph{ego-Facebook},
\emph{Email-Enron},
\emph{email-Eu-core} and 
\emph{Wiki-Vote}.
The number of nodes and the edges of the graphs are shown in Table~\ref{tab:dataset}.

Instead of the above graphs, their subgraphs have been considered. First fixed number of vertices from the graph datasets and edges joining them have been chosen for the subgraphs. For example, for 1000, vertices with identifier $<1000$ have been taken for the subgraph.

\subsection{Experiment Results}
In our experiment, five datasets have been taken. The experiment has been done for each dataset taking extracted subgraphs with vertices 50 to 1000 incremented by 50. The number of edges in the subgraphs is shown in Fig.~\ref{fig:subgraphInfo}.
For the pairing, 128, 256 and 512 bits prime-pairs are taken. 
\begin{figure}[htbp]
	\centering
	{\includegraphics[width=0.48\textwidth]{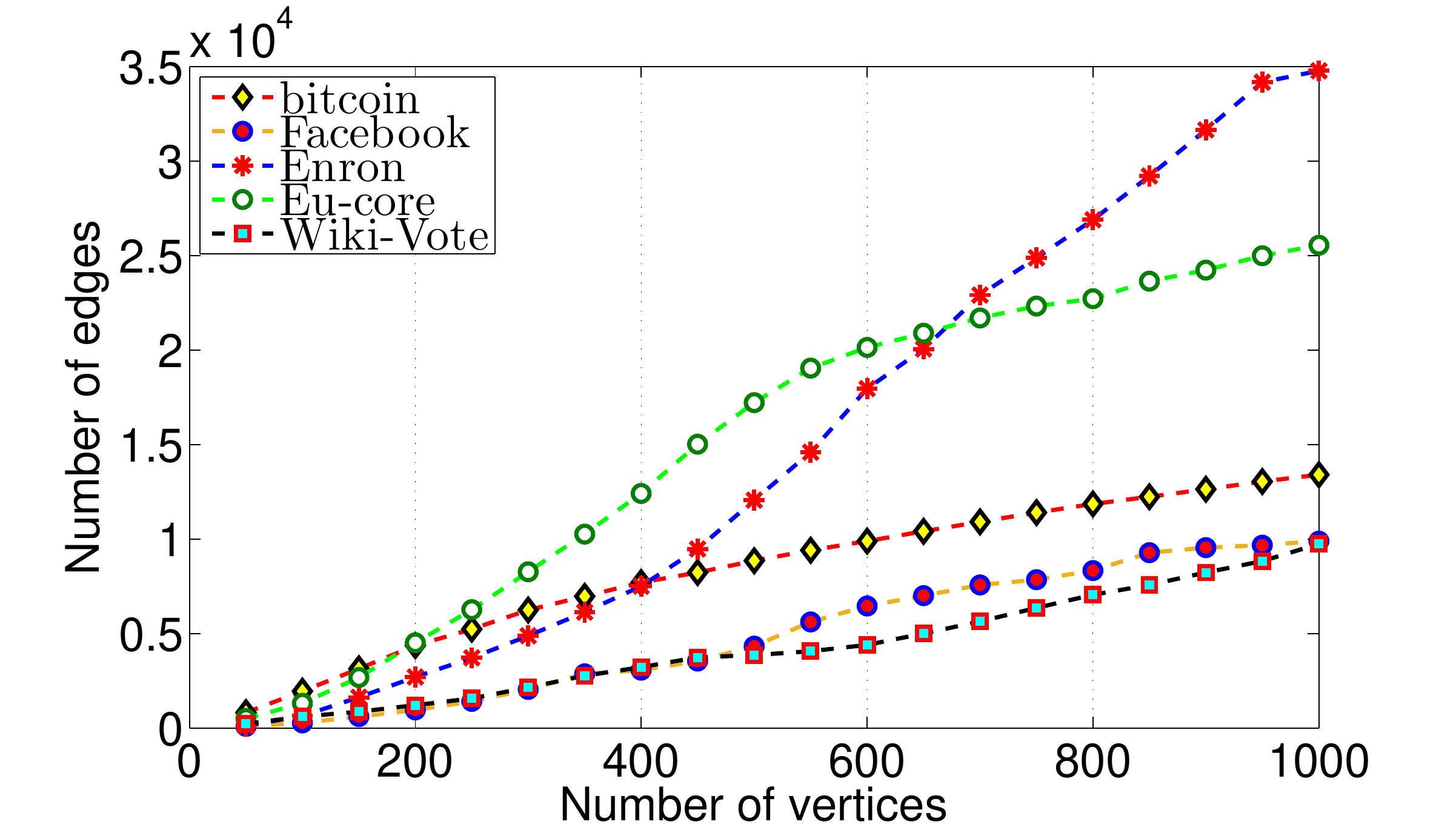}}
	\caption{Number of vertices and edges of the subgraphs}
	\label{fig:subgraphInfo}
\end{figure}
In our proposed schemes, the most expensive operation for the client is encrypting the matrix ($\mathtt{EncMatrix}$). For the cloud and the proxy, score computing ($\mathtt{LPQuery}$) and finding maximum vertex ($\mathtt{FindMaxVertex}$) are the most expensive operations respectively. Hence, throughout this section, we have discussed mainly these three operations.

As we have seen, in the proposed protocols, encrypting each entry of the adjacency matrix is the main operation of the encryption, the number of edges does not affect the encryption time for both $\mathtt{SLP}$-$\mathtt{I}$ and $\mathtt{SLP}$-$\mathtt{II}$. This is because, irrespective of SLP schemes, the number of operations are independent of number of edges. 

\begin{figure}[ht]
	\centering
	\begin{subfigure}[t]{0.33\textwidth}
		\includegraphics[width=0.98\textwidth]{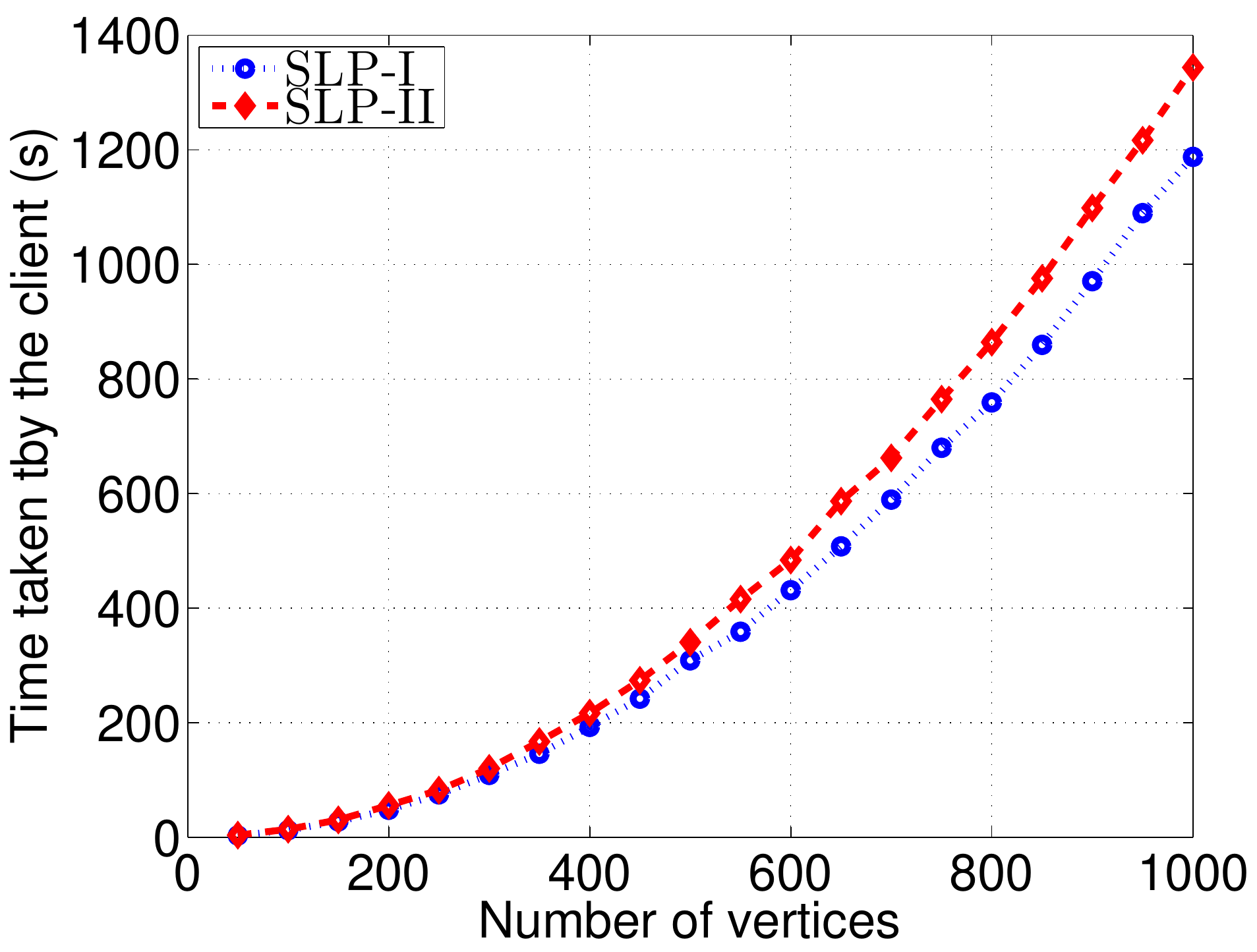}
		\caption{Encryption time taken\\ by the client }
		\label{fig:compClient}
	\end{subfigure}%
	\begin{subfigure}[t]{0.33\textwidth}
		\centering
		\includegraphics[width=0.98\textwidth]{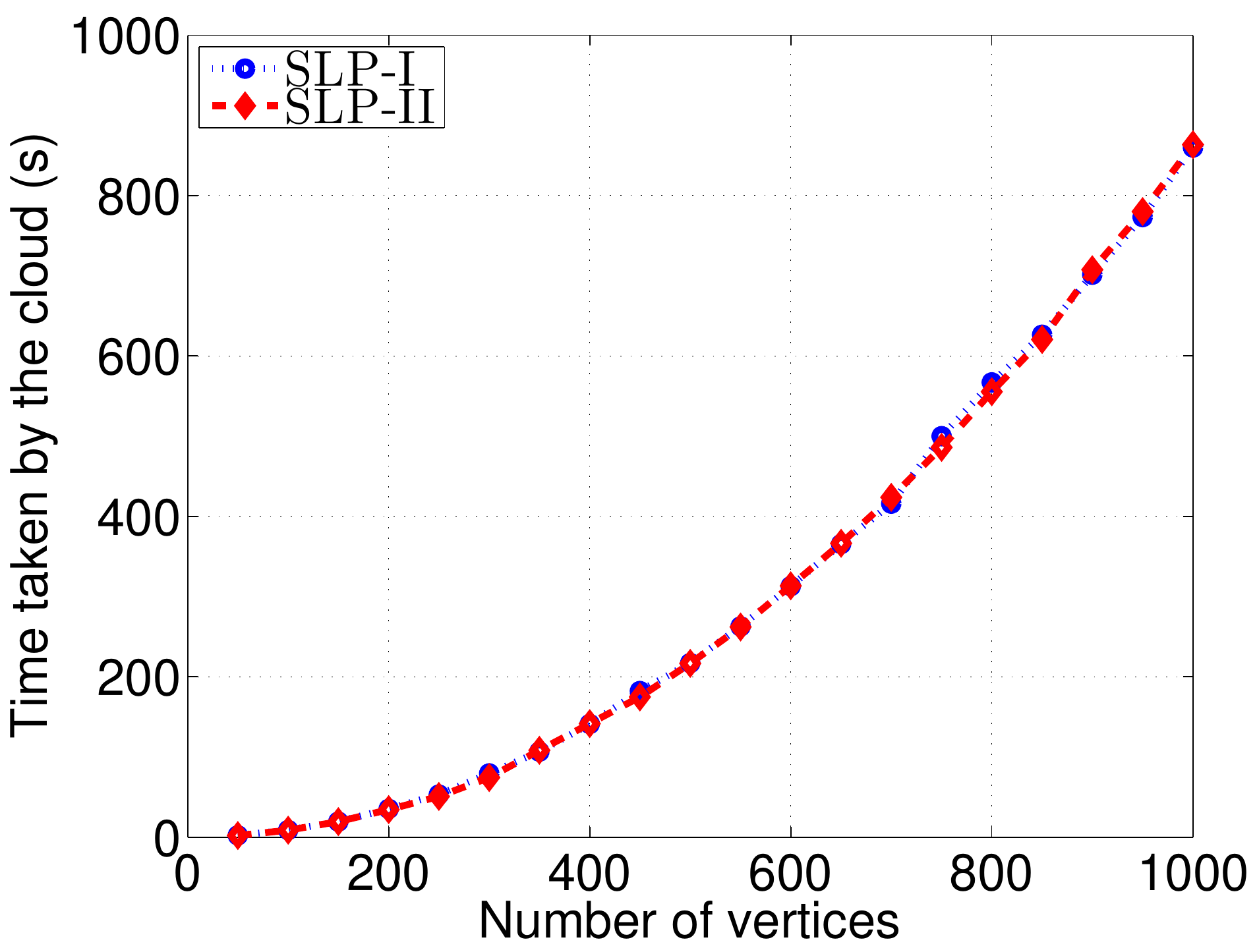}
		\caption{Encrypted score computation times}
		\label{fig:compCloud}
	\end{subfigure}
	\begin{subfigure}[t]{0.33\textwidth}
		\includegraphics[width=0.98\textwidth]{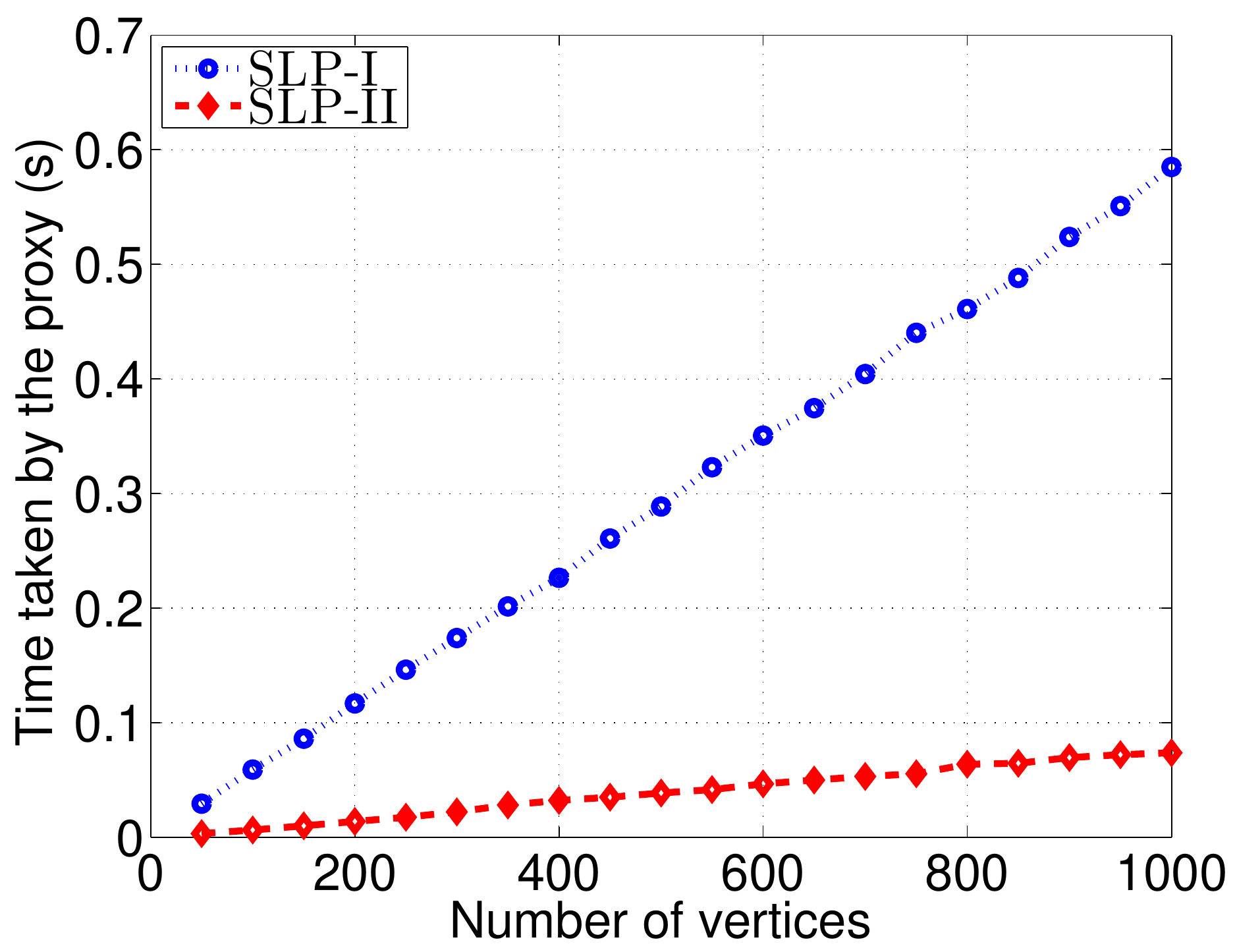}
		\caption{Score decryption and sorting times}
		\label{fig:compProxy}
	\end{subfigure}%
	
	\caption{comparison between $\mathtt{SLP}$-$\mathtt{I}$ and $\mathtt{SLP}$-$\mathtt{II}$ w.r.t. computation time when the primes are of 128 bits each}
	\label{fig:compSLPs}
\end{figure}
Similarly, time required by the cloud to compute score is independent of number of edges and depends on number of entries in the adjacency matrix i.e., $N^2$.
Time taken for each of the operations is shown in Fig.~\ref{fig:compSLPs}. In the figure, we have compared time for both $\mathtt{SLP}$-$\mathtt{I}$ and $\mathtt{SLP}$-$\mathtt{II}$ taking primes 128 bits each.

However, the time taken by the proxy to decrypt the scores is depends on the number of vertices. In  $\mathtt{SLP}$-$\mathtt{I}$, the proxy has to decrypt $|V|$ entries in $\mathbb{G}$ as well as $|V|$ scores in $\mathbb{G}_1$ where in  $\mathtt{SLP}$-$\mathtt{II}$, it decrypts only in $|V|$ scores in $\mathbb{G}_1$. So proxy takes more time in $\mathtt{SLP}$-$\mathtt{I}$ than in $\mathtt{SLP}$-$\mathtt{II}$. This can be observed in Fig.~\ref{fig:compProxy}.

\begin{figure}[!htbp]
	\centering
	{\includegraphics[width=0.32\textwidth]{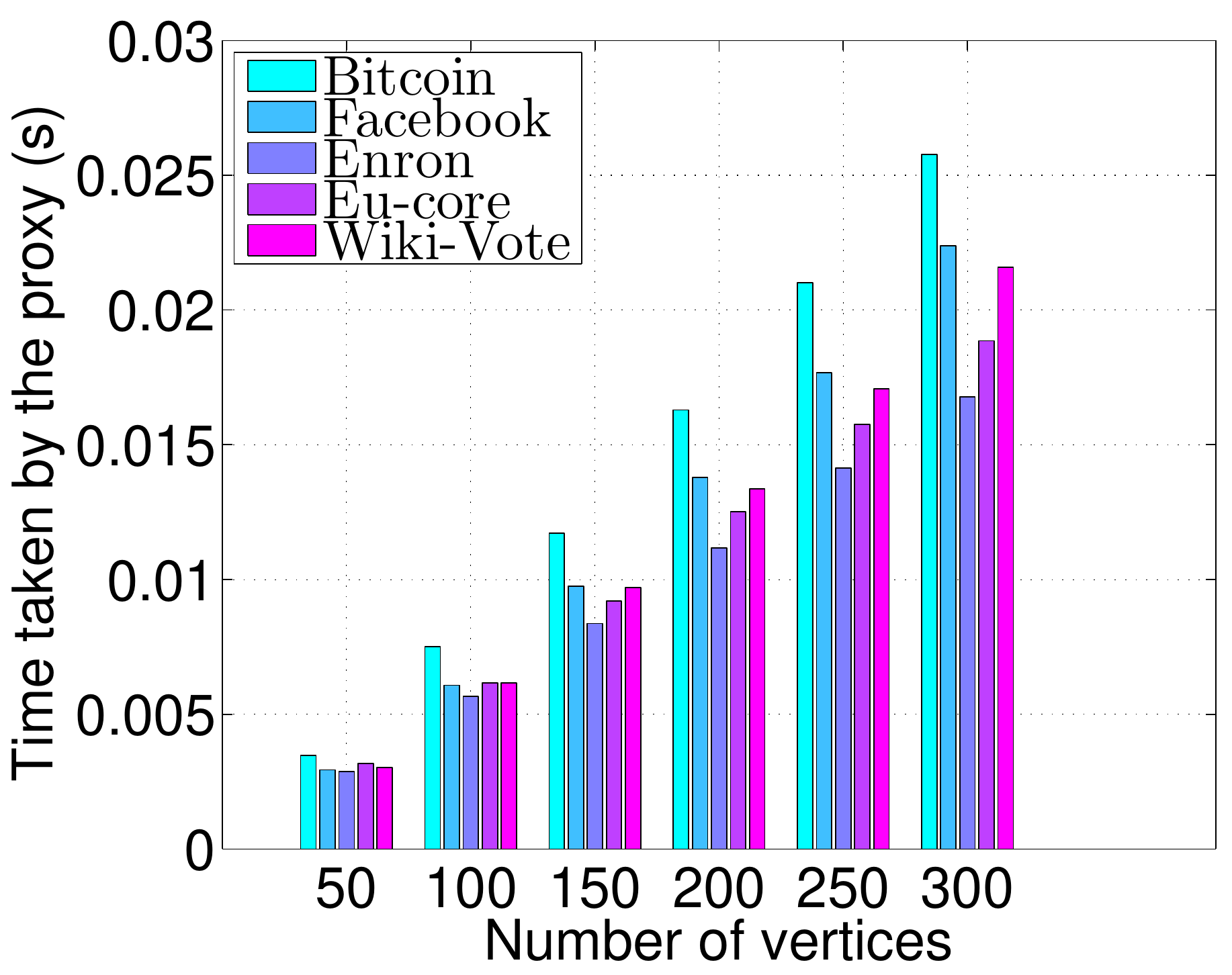}}
	\caption{Time taken by the proxy in $\mathtt{SLP}$-$\mathtt{II}$ for different datasets considering 128-bit primes}
	\label{fig:comp_128_slp2_proxy1}
\end{figure}
For a query, in $\mathtt{SLP}$-$\mathtt{II}$, the proxy decrypts scores only for corresponding vertices that are not incident to the  vertex queried for. So, only in this case, the computational time depends on the number of edges in the graph. As density of edges in a graph increases the chance of decreasing computational time for the graph increases. In Fig.~\ref{fig:comp_128_slp2_proxy1} we have compared computational time taken by the proxy in $\mathtt{SLP}$-$\mathtt{II}$ for different datasets.

In the above figures, we have considered only 128-bit primes. It can be observed from the experiment, the computational time depends on the security parameter. As we increase the size of the primes, the computational time grows exponentially. We have compared the change of computational time for all of the client, cloud and proxy for both  $\mathtt{SLP}$-$\mathtt{I}$  and  $\mathtt{SLP}$-$\mathtt{II}$ (see Fig.~\ref{fig:compTimeSLP1} and Fig.~\ref{fig:compTimeSLP2} respectively).     
However, in practical, as we keep the security bit fixed, keeping the security bits as low as possible improves the performance.

\begin{figure}[ht]
	\centering
	\begin{subfigure}[t]{0.33\textwidth}
		\centering
		\fbox{\includegraphics[width=0.9\textwidth]{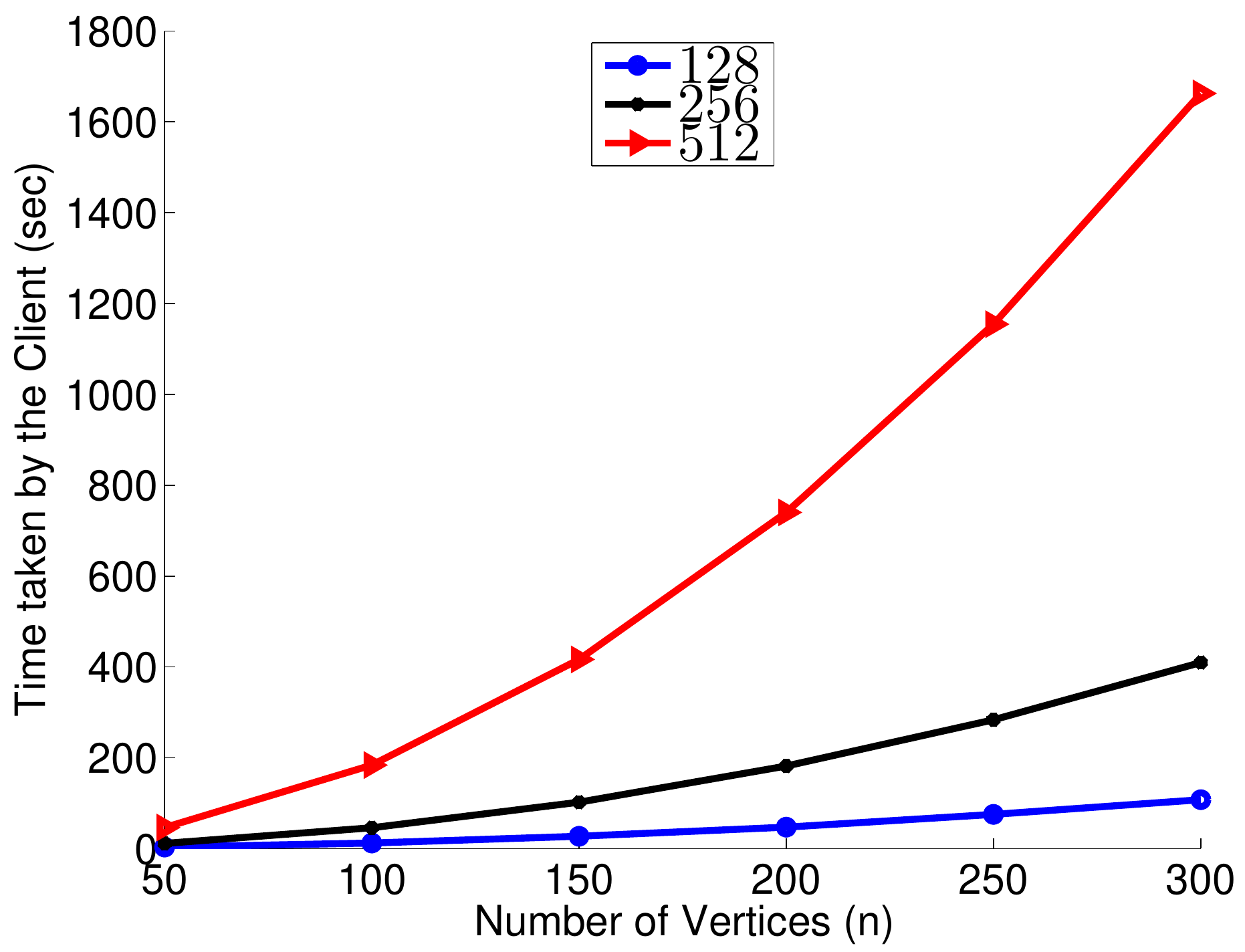}}
		\caption{Client time in $\mathtt{SLP}$-$\mathtt{I}$}
		\label{fig:slp1_client1}
	\end{subfigure}%
	\begin{subfigure}[t]{0.33\textwidth}
		\centering
		\fbox{\includegraphics[width=0.9\textwidth]{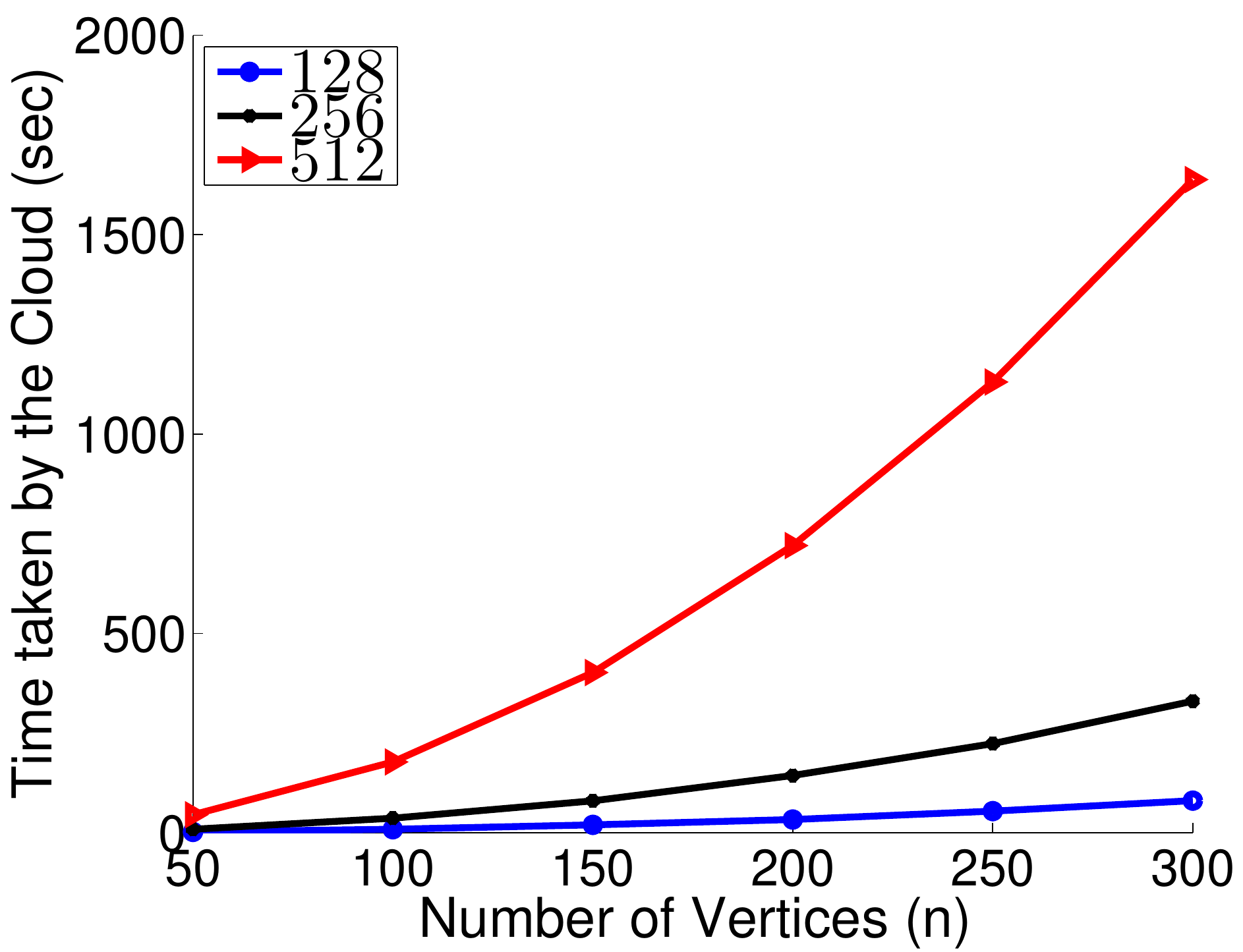}}
		\caption{Cloud time in $\mathtt{SLP}$-$\mathtt{I}$}
		\label{fig:slp1_cloud1}
	\end{subfigure}
	\begin{subfigure}[t]{0.33\textwidth}
		\centering
		\fbox{\includegraphics[width=0.9\textwidth]{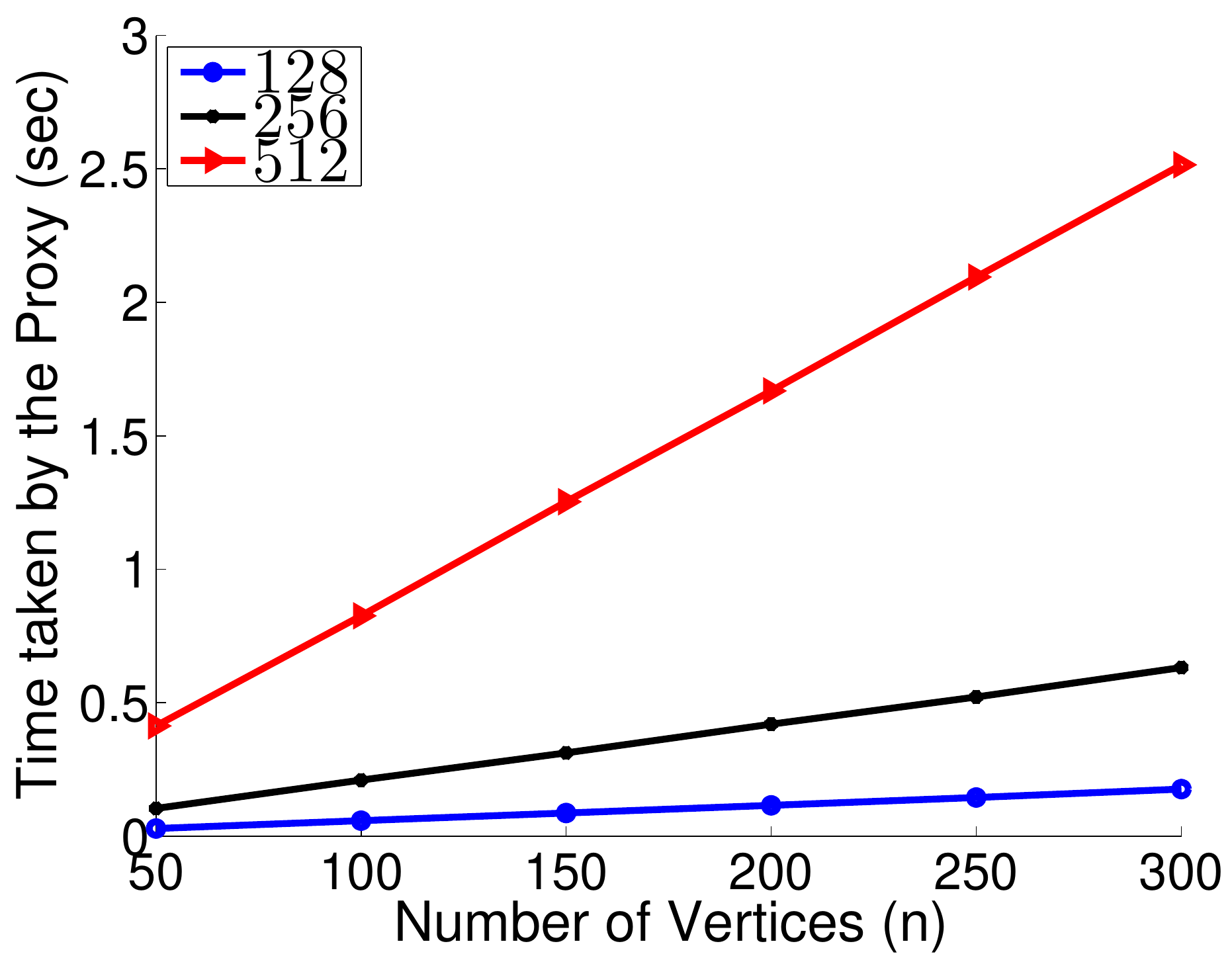}}
		\caption{Proxy time in $\mathtt{SLP}$-$\mathtt{I}$}
		\label{fig:slp1_proxy1}
	\end{subfigure}%
	
	\caption{Computational time in $\mathtt{SLP}$-$\mathtt{I}$ with 128, 256 and 512-bit primes}
	\label{fig:compTimeSLP1}
\end{figure}

\begin{figure}[ht]
	\centering
	\begin{subfigure}[t]{0.33\textwidth}
		\centering
		\fbox{\includegraphics[width=0.9\textwidth]{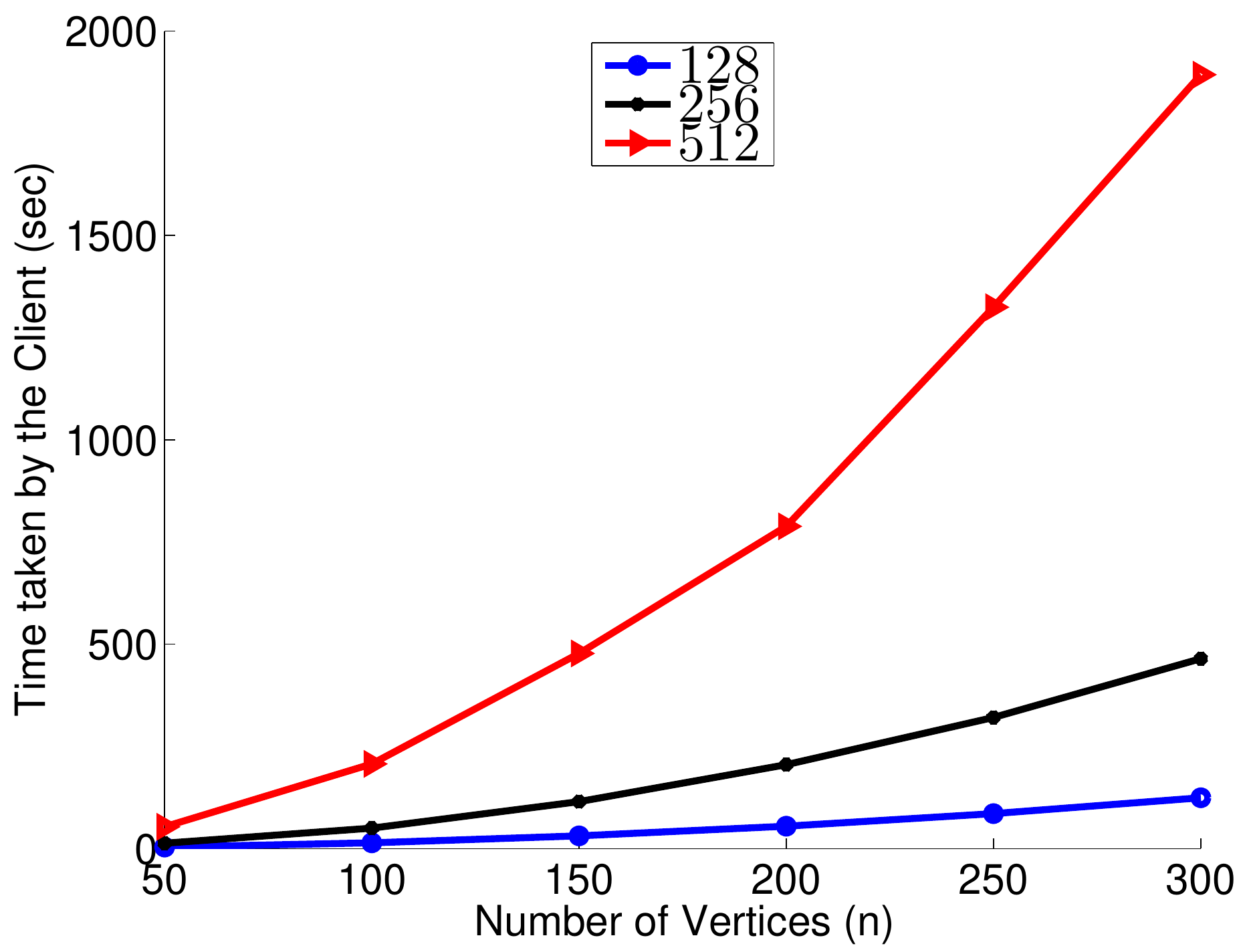}}
		\caption{Client time in $\mathtt{SLP}$-$\mathtt{II}$}
		\label{fig:slp2_client1}
	\end{subfigure}%
	\begin{subfigure}[t]{0.33\textwidth}
		\centering
		\fbox{\includegraphics[width=0.9\textwidth]{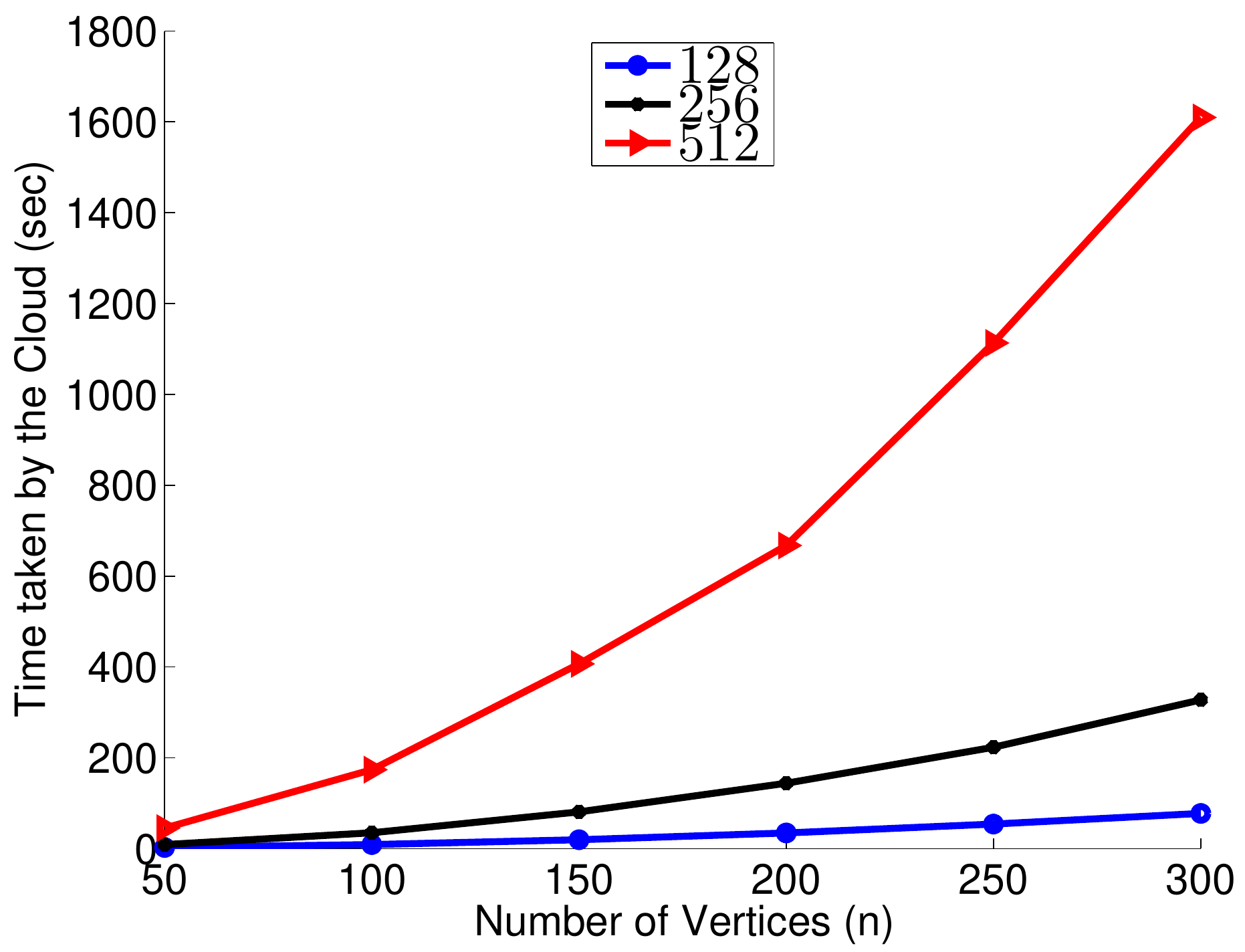}}
		\caption{Cloud time in $\mathtt{SLP}$-$\mathtt{II}$}
		\label{fig:slp2_cloud1}
	\end{subfigure}
	\begin{subfigure}[t]{0.33\textwidth}
		\centering
		\fbox{\includegraphics[width=0.9\textwidth]{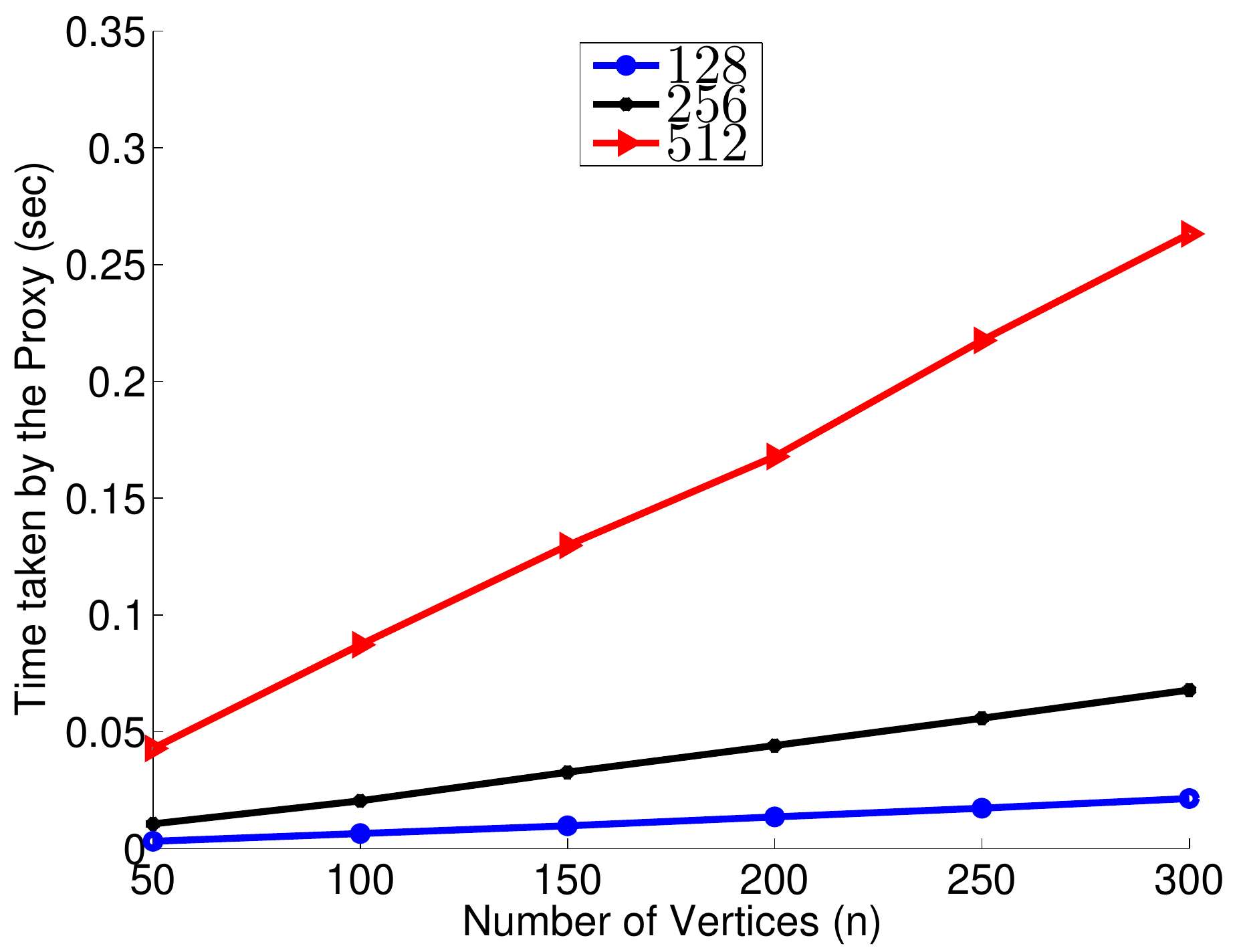}}
		\caption{Proxy time in $\mathtt{SLP}$-$\mathtt{II}$}
		\label{fig:slp2_proxy1}
	\end{subfigure}%
	
	\caption{Computational time in $\mathtt{SLP}$-$\mathtt{II}$ with 128, 256 and 512-bit primes}
	\label{fig:compTimeSLP2}
\end{figure}

\subsection{Estimation of computational cost in $\mathtt{SLP}$-$\mathtt{III}$ } \label{ComputationCost}
In the previous section, we have shown the experimental results for $\mathtt{SLP}$-$\mathtt{I}$ and $\mathtt{SLP}$-$\mathtt{II}$. In this section, we have estimated the computational cost for  $\mathtt{SLP}$-$\mathtt{III}$.
Encryption algorithm of $\mathtt{SLP}$-$\mathtt{III}$ is same as $\mathtt{SLP}$-$\mathtt{I}$. So both required same amount of time for encryption for the same dataset.
To estimate query time, we have considered a random graph with $10^3$ vertices.

\medskip
\noindent
{\bf Query Time:} In $\mathtt{SLP}$-$\mathtt{III}$ the cloud computes encrypted scores and the proxy decrypts the scores as well as random numbers. The number of decryption in each group is same as $\mathtt{SLP}$-$\mathtt{I}$. However, in $\mathtt{SLP}$-$\mathtt{III}$, it requires an extra garbled circuit computation. For this, $1000$ OT for 128-bit security of ECC is required which takes $138*1000$ms = $138$s aprx. (\cite{AsharovL0Z13,NaorP01}). In addition to that, the PS evaluates the GC with $1000*(11*257+4)=2831000$ XOR-gates and $1000*(5*257+1)=1286000$ AND-gates. Assuming that the encryption used in each GC circuit is AES (128-bit), GC evaluation requires 2 AES decryption and the CS requires 8 encryption. As we see in \cite{benchmarkLink}, it requires 0.57 cycles per byte for AES. Thus, for evaluation in a single core processor, the PS requires (2*(1286000*256/8)*0.57) cycles = 46913280 cycles that takes $(46913280/(2.5*10^9))= 0.019$s. Similarly, The CS requires 0.078s to construct the GC.

The estimated costs are measured with respect to a single core 2.5 GHz processor. However, in practice, the CS provides a large number of multi-core processors. As we see all the computations can be computed in \emph{parallel}, the query cost can be reduced dramatically. Each of the above-mentioned costs can be improved to $\frac{cost}{p}$s with $p$ processors and cost is $cost$.

\section {Introduction to $\mathtt{SLP}_k$}\label{sec:slpk}
Let us define another variant of secure link prediction problem $\mathtt{SLP}_k$. Instead of returning the vertex with highest score, an $\mathtt{SLP}_k$ returns indices of $k$ number of top-scored vertices.

Let, a graph $G = (V,E)$ is given. Then, the \emph{top-$k$ Link Prediction Problem} states that given a vertex $v \in V$, it returns a set of vertices $\{u_1, u_2, \ldots, u_k \}$ such that $score(v,u_i) $ is among top-$k$ elements in $S_v$.
The top-$k$ link prediction scheme is said to be secure i.e., a secure top-$k$ link prediction problem scheme ($\mathtt{SLP}_k$) if, the servers do not get any meaningful information about $G$ from its encryption or sequence of queries. 

Our proposed schemes, $\mathtt{SLP}$-$\mathtt{I}$ and $\mathtt{SLP}$-$\mathtt{II}$, can be extended to support $\mathtt{SLP}_k$ queries. In $\mathtt{SLP}$-$\mathtt{I}$, the only change is that instead of returning only the index of the vertex with highest score, the proxy has to return the indices of the top-$k$ highest scores to the client.

\section{Conclusion} \label{sec:Conclusion} 

In this paper, we have introduced the secure link prediction problem and discussed its security. We have presented three constructions of SLP. 
The first proposed scheme $\mathtt{SLP}$-$\mathtt{I}$ has the least computational time with maximum leakage to the proxy. The second one $\mathtt{SLP}$-$\mathtt{II}$ reduces the leakage by randomizing scores. However, it suffers high communication cost from proxy to the client. The third scheme $\mathtt{SLP}$-$\mathtt{III}$ has minimum leakage to the proxy. Though the garbled circuit helps to reduce leakage, it increases the communication and computational cost of the cloud and the proxy servers.  

Performance analysis shows that they are practical. We have implemented prototypes of first two schemes and measured the performance by doing experiment with different real-life datasets. We also estimated the cost for $\mathtt{SLP}$-$\mathtt{III}$. In the future, we want to make a library that support multiple queries including neighbor query, edge query, degree query, link prediction query etc.  

It is to be noted that the cost of computation without privacy and security is far better. The performance has been degraded since we have added security. The performance comes at the cost of security.

Throughout the paper, we have considered unweighted graph. As a future work the schemes can be extended to weighted graphs.
Moreover, we have initiated the secure link prediction problem and considered only common neighbors as score metric. As a future work, we will consider the other distance metrics like Jaccard's coefficient, Adamic/Adar, preferential attachment, Katz$_\beta$ etc. and compare the efficiency of each.

\section*{Acknowledgments} We thank Gagandeep Singh and Sameep Mehta of IBM India research for their initial valuable comments on this work.



\end{document}